\providecommand{\forlics}[1]{}
\providecommand{\podc}[1]{}
\providecommand{\disc}[1]{} %
\providecommand{\arxiv}[1]{#1}

\arxiv{
\documentclass[12pt]{article}
\usepackage{fullpage}
}
\podc{
\documentclass[manuscript,nonacm,anonymous]{acmart}
\setcopyright{none}
\copyrightyear{2020}
\acmYear{2020}
\acmDOI{}
\acmConference[]{}{}{}
\acmBooktitle{}
\acmPrice{}
\acmISBN{}
\pdfoutput=1
\settopmatter{printacmref=false} %
\renewcommand\footnotetextcopyrightpermission[1]{} %
}
\disc{
\documentclass[a4paper,anonymous,USenglish]{lipics-v2021}
\bibliographystyle{plainurl}%
}
\forlics{
\documentclass[format=sigplan, review=true, screen=true]{acmart}
}

\usepackage{tikz}
\usetikzlibrary{arrows,shapes}
\usetikzlibrary{arrows.meta}
\usetikzlibrary{positioning,shapes,shadows,arrows,calc}
\usepackage{url}
\usepackage{xspace} 
\usepackage{alltt} %
\usepackage[T1]{fontenc}
\usepackage{amsmath}
\usepackage{amssymb} %
\usepackage{amsthm} %
\usepackage{caption,subcaption} %
\usepackage{csquotes}
\usepackage{hyperref}
\usepackage{listings}
\usepackage{mathtools}
\usepackage{tcolorbox}
\usepackage{dsfont} %
\usepackage{multirow}
\usepackage{pgfplots}

\usepackage{mathrsfs} %
\newtheorem{proposition}{Proposition}
\newtheorem{theorem}{Theorem}
\newtheorem{corollary}{Corollary}[proposition]
\lstdefinestyle{mystyle}{
    backgroundcolor=\color{white},   
    basicstyle=\footnotesize\ttfamily, %
    commentstyle=\color{brown},
    keywordstyle=\color{magenta}, %
    stringstyle=\color{teal},
    numberstyle=\tiny\color{codegray},
    breakatwhitespace=false,
    breaklines=false,                 
    captionpos=b,
    keepspaces=true,                 
    numbers=left,
    numbersep=5pt,
    showspaces=false,                
    showstringspaces=false,
    showtabs=false,                  
    tabsize=2,
    numbers=none, %
    xleftmargin=1ex,
    frame = single,
    escapeinside={(*}{*)} %
}
\usepackage{mycmds}

\newcommand\Raw{{\bf Raw}\xspace}
\newcommand\Fair{{\bf Fair}\xspace}
\newcommand\Sure{{\bf Sure}\xspace}
\newcommand\AlwQ{{\bf Alw-Q}\xspace}
\newcommand\QAlw{{\bf Q-Alw}\xspace}
\newcommand\PAlwQ{{\bf P-Alw-Q}\xspace}
\newcommand\PQAlw{{\bf PQ-Alw}\xspace}
\newcommand\Alw{{\bf Alw}\xspace}
\newcommand\PQDur{{\bf PQ-Dur}\xspace}
\newcommand\PQExtraDur{{\bf PQ-Extra-Dur}\xspace}
\newcommand\EachVote{{\bf Each-Vote}\xspace}
\newcommand\SomeLearn{{\bf Some-Learn}\xspace}
\newcommand\EachLearn{{\bf Each-Learn}\xspace}
\newcommand\SomeExec{{\bf Some-Exec}\xspace}
\newcommand\EachExec{{\bf Each-Exec}\xspace}
\newcommand\Resp{{\bf Resp}\xspace}

\providecommand{\msg}[1]{\textcolor{teal}{\co{#1}}}
\newcommand{\kw}[1]{\textcolor{magenta}{\co{#1}}}
\newcommand\AT{\kw{at}\xspace}
\newcommand\ALW{\kw{alw}\xspace}
\newcommand\EVT{\kw{evt}\xspace}
\newcommand\EACH{\kw{each}\xspace}
\newcommand\SOME{\kw{some}\xspace}
\newcommand\HAS{\kw{has}\xspace}
\newcommand\DURING{\kw{during}\xspace}
\newcommand\LASTS{\kw{lasts}\xspace}
\newcommand\AFTER{\kw{after}\xspace}

\arxiv{
\title{What's~Live? Understanding~Distributed~Consensus\thanks{\fund}}
}
\disc{
\title{What's~Live?~Understanding~Distributed~Consensus}
}
\podc{
\title[What's Live? Understanding Distributed Consensus]{
What's Live? Understanding Distributed Consensus}
}

\newcommand\fund{
  This work was supported in part by NSF under grants 
  CCF-1414078 and %
  CCF-1954837 %
  and ONR under grant
  N000142012751. %
}

\arxiv{
\author{Saksham Chand \Hex{8} Yanhong A. Liu\Vex{1}\\
  {\small Computer Science Department, Stony Brook University,
  Stony Brook, NY 11794, USA}\\
  {\footnotesize \tt\{schand,liu\}@cs.stonybrook.edu}\Vex{-2}}
\date{}
}
\podc{
\author{Saksham Chand}
\email{schand@cs.stonybrook.edu}
\author{Yanhong A. Liu}
\email{liu@cs.stonybrook.edu}
\affiliation{%
  \department{Computer Science Department}
  \institution{Stony Brook University}
  \city{Stony Brook}
  \state{New York}
  \postcode{11794}
  \country{USA}
}
}
\disc{
\author{Saksham Chand}{Stony Brook University}{schand@cs.stonybrook.edu}{}{}
\author{Yanhong A. Liu}{Stony Brook University}{liu@cs.stonybrook.edu}{}{}

\Copyright{Saksham Chand and Yanhong A. Liu}



\category{} %

\relatedversion{} %

\EventEditors{John Q. Open and Joan R. Access}
\EventNoEds{2}
\EventLongTitle{42nd Conference on Very Important Topics (CVIT 2016)}
\EventShortTitle{CVIT 2016}
\EventAcronym{CVIT}
\EventYear{2016}
\EventDate{December 24--27, 2016}
\EventLocation{Little Whinging, United Kingdom}
\EventLogo{}
\SeriesVolume{42}
\ArticleNo{23}
}

\begin{document}
\maketitle

\notes{ %
We show that all properties specified that too weak, either assuming there will be no failures ...
We further propose a stronger property than all these properties, and show that Paxos Made Moderately Complex, when extended with failure detection, satisfies this property.
} 
\providecommand{\abstractText}{%

Distributed consensus algorithms such as Paxos have been studied
extensively. They all use the same definition of safety. Liveness is
especially important in practice despite well-known theoretical
impossibility results. However, many different liveness properties and
assumptions have been stated, and there are no systematic comparisons for
better understanding of these properties.

This paper systematically studies and compares different liveness
properties stated for over 30 prominent consensus algorithms and variants.
We introduce a precise high-level language and formally specify these
properties in the language.
We then create a hierarchy of liveness properties combining two hierarchies
of the assumptions used and a hierarchy of the assertions made, and compare
the strengths and weaknesses of algorithms that ensure these properties.
Our formal specifications and systematic comparisons led to the discovery
of a range of problems in various stated liveness properties, from too weak
assumptions for which no liveness assertions can hold, to too strong
assumptions making it trivial to achieve the assertions.
We also developed TLA+ specifications of these liveness properties, and we
use model checking of execution steps to illustrate liveness patterns for
Paxos. 
}

\begin{abstract}
\abstractText
\end{abstract}

\pagestyle{plain}
\thispagestyle{plain}
\hypersetup{
    linkcolor=red,
    citecolor=green,
    filecolor=cyan,
    menucolor=red,
    urlcolor=cyan,
}

\section{Introduction}

Distributed systems are pervasive, where processes work with each other via
message passing.  For example, this article is available to the current
reader owing to some distributed system.  At the same time, distributed
systems are highly complex, due to their unpredictable asynchronous nature
in the presence of failures---processes can fail and later recover, and
messages can be lost, delayed, reordered, and duplicated.
This paper considers the critical problem of distributed consensus---a set of processes trying to agree on a value or a continuing sequence of values%
---under these possible failures.

Distributed consensus under these possible failures is essential in
services that must maintain a state, including many services provided by
companies like Google with, e.g., the Chubby distributed lock
service~\cite{burrows2006chubby} and the BigTable distributed storage
system~\cite{chang2008bigtable}.
This is because such services must use replication to tolerate failures
caused by machine crashes, network outages, etc.  Replicated servers must
agree on the state of the service or the sequence of operations that have
been performed---for example, a customer order has been placed and paid but
not yet shipped---so that when some servers become unavailable, the
remaining servers enable the service to function correctly.

Many algorithms and variations have been created for this consensus
problem, including the most prominent ones listed in
Table~\ref{tab-algo}. These algorithms share similar ideas, with
Paxos~\cite{lamport1998part} becoming the most well-known name and the
focus of many studies.
These algorithms are required to be safe in that only a single value or a
single sequence of values have been agreed on at any time, and the value or
values agreed on are among the values proposed by the processes.
Proving safety of these algorithms has been of utmost interest 
since time immemorial---with varying degrees of rigor, from informal proof
sketches to fully formal machine-checked proofs.

However, liveness of these algorithms---certain desired progress has been
made in reaching agreement---is not always given the same kind of
attention. For example, out of the 31 algorithms listed in
Table~\ref{tab-algo}, only 18 (58\%) of them mention liveness
(Table~\ref{tab-live-summary}) and only 11 (35\%) of them give some kind of
liveness proof (Column ``Proofs'' of Table~\ref{tab-live-summary}).  This
is intriguing because real-world implementations of consensus algorithms
are required to be live,
a.k.a.\ ``responsive'', despite the FLP impossibility
results~\cite{fischer1985impossibility}.  Furthermore, prima facie, it is sometimes
unclear or ambiguous what exactly liveness means for these algorithms and
under what exact conditions they are claimed to satisfy those properties.

\mypar{This paper}
This paper considers consensus algorithms and variants listed in
Table~\ref{tab-algo} and examines and assembles their liveness properties
as stated in the literature, 
including the assumptions under which the properties are satisfied.
We formally specify and precisely compare these liveness properties 
and present our analysis for all of them (summarized in
Table~\ref{tab-live-summary}).
We categorize the assumptions into two classes: (1) link assumptions,
stating what is assumed about the communication links, and (2) server
assumptions, stating what is assumed about the servers in the system.  

We identify three different kinds of link assumptions, seven different
kinds of server assumptions, and six different kinds of liveness assertions
capturing different kinds of progress.
We relate different assumptions and assertions by creating three
hierarchies: a 3-element total order for link assumptions, a 7-element
partial order for server assumptions that includes two diamonds, and a
6-element partial order for liveness assertions that includes a diamond in
the middle.
Together they form an overall hierarchy of liveness properties, which
includes all the liveness properties stated for the algorithms and variants
considered.

We introduce a precise high-level language and formally specify all these
assumptions and assertions in the language.  In fact, it was the precise
language and formal specifications that allowed us to establish the precise
relationships among all the assumptions and assertions.
The overall hierarchy
not only helps in understanding liveness better, but also helps in finding
tighter results.
It led us to the discovery of a range of problems in various stated
liveness properties, from lacking assumptions or too weak assumptions under
which no liveness assertions can hold, to too strong assumptions making it
trivial or uninteresting to satisfy the liveness assertions.  For example,
\begin{itemize}
\item EPaxos~\cite{moraru2013there} only assumes that eventually, there is
  always some quorum of servers that is non-faulty (what we call \AlwQ),
  but this is too weak for reaching consensus, because each such quorum
  might not stay non-faulty long enough to make progress.

\item Zab~\cite{junqueira2011zab}, Paxso-EPR~\cite{padon2017paxos}, and
  some others assume that eventually, there is a server P and a quorum Q of
  servers, such that always, P is non-faulty and is the primary and Q is
  non-faulty (what we call \PQAlw).  Chubby-Live~\cite{chandra2007paxos}
  even assumes that eventually, all servers are always non-faulty.  Such
  strong assumptions make it trivial to reach consensus.

\end{itemize}
This also led us to the question of what the right or best liveness
properties are among the many possible combinations of assumptions and
assertions.  In fact, many protocols assert stronger properties than
necessary.

We also developed TLA+ specifications of these liveness properties. We
show, in Appendix~\ref{app-mc}, that model checking execution steps using
TLC can illustrate liveness patterns for single-value Paxos on up to 4
proposers and 4 acceptors in a few hours, but becomes too expensive for
multi-value Paxos or more processes. 
Appendices~\ref{appendix-live-vs}, \ref{appendix-live-vr},
\ref{appendix-live-paxos}, \ref{appendix-live-fd},
and~\ref{appendix-live-others} show the liveness properties and assumptions
statements that we extracted and assembled from all the algorithms and
variants considered.

\section{Preliminaries}
\label{sec-preliminaries}

This section introduces the terminology used in this paper, the algorithms
and variants considered, and the high-level structure of these algorithm
and variants.

\mypar{Distributed systems and failures considered}
A distributed system is a set of processes each operating on its local data
and communicating with other processes by sending and receiving messages.
A process may crash and may later recover, and a message may be lost,
delayed, put out of order, and/or duplicated. A process is said to be
\defn{non-faulty} if ``it eventually performs the actions that it should,
such as responding to messages''~\cite{lamport2006fast}.

\mypar{Distributed consensus}
The basic consensus problem, called single-value consensus, is to ensure
that a single value is chosen by non-faulty processes from among the set of
values proposed by the processes. The more general consensus problem,
called multi-value consensus, is to choose a single sequence of values,
instead of a single value.

\mypar{Terminology} 
We identify two types of processes: \defn{servers}, which execute some
consensus algorithm and together provide consensus as a service, and
\defn{clients}, which use this service. In case of multi-value consensus,
clients send values to servers for them to order, and each server maintains
a \defn{log}, which is a sequence of chosen values.  The indices of entries in the log
are called \defn{slots}.

A server is said to be \defn{primary} if it is explicitly elected by the
algorithm to be the single primary or leader among all servers in the
system. Some other names for primary that appear in the literature are
president~\cite{lamport1998part},
leader~\cite{lamport2006fast,prisco00revisit,kirsch2008paxos,junqueira2011zab,ongaro2014search},
and coordinator~\cite{mao2008mencius,druagoi2016psync}. In general, the
defining property of the primary is that it is the only server that can
propose values to other servers. Note that some consensus algorithms do not use
a primary.

A \defn{quorum system} $Qs$ is a set of subsets, called quorums, of servers
such that any two quorums overlap, that is,
$\forall Q_1, Q_2 \in Qs : Q_1 \cap Q_2 \neq \emptyset$. The most commonly
used quorum system $Qs$ takes any majority of servers as an element in
$Qs$. A quorum is said to be \textit{non-faulty} if each of its servers is
non-faulty, and \textit{faulty} otherwise.

\notes{
algorithms: vr, paxos, ...
discuss some grouping: 
passive replication vs active replication - sc: we can discuss primary vs non-primary based consensus, because very few of these algorithms actually carry out replication. Most are limited to consensus.

variants: as specified in various formal and informal settings.
executable, or not.
}

\disc{
\providecommand{\tableAlgoA}{0.13}
\providecommand{\tableAlgoB}{0.685}
\providecommand{\tableAlgoC}{0.23}
}
\podc{
\providecommand{\tableAlgoA}{0.12}
\providecommand{\tableAlgoB}{0.61}
\providecommand{\tableAlgoC}{0.2}
}
\arxiv{
\providecommand{\tableAlgoA}{0.135}
\providecommand{\tableAlgoB}{0.69}
\providecommand{\tableAlgoC}{0.16}
}
\forlics{
\providecommand{\tableAlgoA}{0.1}
\providecommand{\tableAlgoB}{0.6}
\providecommand{\tableAlgoC}{0.2}
}
\providecommand{\tableAlgo}{%
\begin{tabular}{@{\arxiv{\Hex{-3}}\disc{\Hex{-3}}\podc{~~}}r@{~}|@{~}p{\tableAlgoA\textwidth}@{\,}|@{~}p{\tableAlgoB\textwidth}@{\,}|@{~}p{\tableAlgoC\textwidth}@{}}
~~  & Name & Description & Language used %
\\
\hline\hline
\rownumber & VS-ISIS & Reliable group communication, Birman-Joseph 1987~\cite{birman1987reliable} & English (items)\\
\rownumber & VS-ISIS2 & Virtual synchrony, Birman-Joseph 1987~\cite{birman1987exploiting} & English\\
\rownumber & EVS & Extended virtual synchrony for network partition, Amir et al 1995~\cite{amir1995replication,amir1995totem} & pseudocode\\
\rownumber & Paxos-VS & Virtually synchronous Paxos, Birman-Malkhi-van Renesse 2012~\cite{birman2010virtually} & pseudocode\\
\rownumber & Derecho & Virtually synchronous state machine replication, Jha et al 2019~\cite{jha2019derecho} & pseudocode\\
\hline\hline

\rownumber & VR & Viewstamped replication, Oki-Liskov 1988~\cite{oki1988viewstamped} & pseudocode (coarse)\\
\rownumber & VR-Revisit & VR revisited, Liskov 2012~\cite{liskov2012viewstamped} & English (items)\\
\hline\hline

\rownumber & Paxos-Synod & Paxos in part-time parliament, Lamport 1998~\cite{lamport1998part} & TLA\,\cite{lamport1994temporal} (single-value)\\
\rownumber & Paxos-Basic & Single-value Paxos, Lamport 2001~\cite{lamport2001paxos} & English (items)\\
\rownumber & Paxos-Fast & Single-value Paxos with replicas proposing, Lamport 2006~\cite{lamport2006fast} & English\,\,(items), TLA+\,\cite{chaudhuri2010tla+}\\
\rownumber & Paxos-Vertical & Single-value Paxos with external starting of leader election,
Lamport-Malkhi-Zhou 2009~\cite{lamport2009vertical} & PlusCal~\cite{lamport2009pluscal}\\
\hline\hline

\rownumber & CT & Single-value consensus with crash failures, Chandra-Toueg 1996~\cite{chandra1996unreliable} & pseudocode\\
\rownumber & ACT & Single-value consensus in crash-recovery model, Aguilera-Chen-Toueg 2000~\cite{aguilera2000failure} & pseudocode\\
\hline\hline

\notes{
\rownumber & GIRAF & Liveness in generalized round-based algorithm framework, Keidar-Shraer 2006~\cite{keidar2006timeliness,keidar2006timelinessTR} & pseudocode, IOA\\
\hline\hline
}

\rownumber & Paxos-Time & Paxos with time analysis, De Prisco-Lampson-Lynch 2001~\cite{prisco00revisit} & IOA\,\cite{lynch1988introduction} (single-value)\\

\rownumber & Paxos-PVS & Single-value Paxos for proof, Kellom{\"a}ki 2004~\cite{kellomaki2004annotated} & PVS~\cite{owre1992pvs}\\
\hline

\rownumber & Chubby & Paxos in Google's Chubby lock service, Burrows 2006~\cite{burrows2006chubby} & English (partial items)\\

\rownumber & Chubby-Live & Chubby in Paxos made live, Chandra-Griesemer-Redstone 2007~\cite{chandra2007paxos} & English\\

\rownumber & Paxos-SB & Paxos for system builders, Kirsch-Amir 2008~\cite{kirsch2008paxos} & pseudocode\\

\rownumber & Mencius & Paxos with leaders proposing in turn, Mao et al 2008~\cite{mao2008mencius} & English (items)\\

\rownumber & Zab & Yahoo/Apache's Zookeeper atomic broadcast, Junqueira-Reed-Serafini 2011~\cite{junqueira2011zab} & English (items)\\

\rownumber & Zab-FLE & Zab with fast leader election, %
Medeiros 2012~\cite{medeiros2012zookeeper} & pseudocode\\

\rownumber & EPaxos & Egalitarian Paxos, Moraru-Andersen-Kaminsky 2013~\cite{moraru2013there} & pseudocode\\

\rownumber & Raft & Consensus in RAMCloud, Ongaro-Ousterhout 2014~\cite{ongaro2014search} & pseudocode\\

\rownumber & Paxos-Complex & Paxos made moderately complex, van Renesse-Altinbuken 2015~\cite{van2015paxos} & pseudocode, Python\\
\hline

\rownumber & Raft-Verdi & Raft for proof using Coq, 
Wilcox et al 2015~\cite{wilcox2015verdi} & Verdi~\cite{wilcox2015verdi}\\

\rownumber & IronRSL & Paxos in Microsoft's IronFleet for proof, Hawblitzel et al 2015~\cite{hawblitzel2015ironfleet} & Dafny~\cite{leino2010dafny}\\

\rownumber & Paxos-TLA & Paxos for proof using TLAPS, Chand-Liu-Stoller 2016~\cite{chand2016formal}& TLA+\\

\rownumber & LastVoting-PSync & Single-value Paxos in Heard-Of model for proof, %
Dr{\u{a}}goi-Henzinger-Zufferey 2016~\cite{druagoi2016psync} & PSync~\cite{druagoi2016psync}\\

\rownumber & Paxos-EPR & Paxos in effectively propositional logic for proof, Padon et al 2017~\cite{padon2017paxos} & Ivy~\cite{padon2016ivy}\\

\rownumber & Paxos-Decon & Paxos deconstructed, %
Garcia et al 2018~\cite{garcia2018paxos,garcia2018paxos-arxiv}& Scala/Akka~\cite{akka2016doc}\\

\rownumber & Paxos-High & Paxos in high-level executable specification, Liu-Chand-Stoller 2019~\cite{liu2019moderately} & DistAlgo~\cite{liu2017clarity}\\
\hline
\end{tabular}
}

\newcounter{magicrownumbers}
\newcommand\rownumber{\stepcounter{magicrownumbers}\arabic{magicrownumbers}}
\begin{table*}[htp]
\disc{\Vex{-3}}
\small
\centering
\tableAlgo
\disc{\Vex{-4}}
\caption{Distributed consensus algorithms and variants, and languages used to express them.\podc{\Vex{-4}}\disc{\Vex{-2}}}
\label{tab-algo}
\end{table*}

\mypar{Algorithms and variants}
Table~\ref{tab-algo} lists 31 prominent algorithms and variants for
distributed consensus in the literature, together with the languages in
which they are expressed.
\begin{itemize}

\item The first group (1-5) comprises systems and algorithms based on what
  is known as Virtual Synchrony~\cite{birman1987exploiting}, the first
  algorithm
  for reliable group communication using broadcast primitives.

\item The second group (6-7) is for what is called Viewstamped
  Replication~\cite{oki1988viewstamped}, the earliest algorithm that has a
  primary-backup architecture.

\item The third group (8-11) is for Paxos~\cite{lamport1998part} and core
  variants, the earliest algorithm based on symmetric leader election and
  state-machine replication. Basic Paxos solves the problem of single-value
  consensus. Multi-Paxos, hereafter simply called Paxos, solves the problem
  of multi-value consensus.

\item The fourth group (12-13) is for two core classes of single-value
  consensus algorithms with failure detection, one where servers follow
  crash failure model (12) and the other where servers follow crash-recover
  failure model (13).

\item The fifth group (14-31) is for other algorithm variations. The
  earliest two (14-15) and latest seven (25-31) in this group are variants
  specified formally, in a language with machine-supported syntax and
  semantics checking, and with proof supports or compilation and execution.

\end{itemize}
Note that the similarity and differences among the names do not necessarily
reflect those in the algorithms or variants.  For example, Paxos-VS differs
from Paxos much more than Raft differs from Paxos, whereas Chubby and Raft
are very similar despite having completely different names.

\mypar{High-level algorithm structure} At a (sufficiently) high-level, all
of the algorithms in Table~\ref{tab-algo} are quorum-based consensus
algorithms that execute in \defn{rounds}. A round, also known as a ballot,
view, epoch, and term, is identified by a unique number, and is initiated
by at most one server. Multiple rounds may run simultaneously. 
\begin{itemize}

\item A server initiates a new round once it thinks that the current round
  has died, i.e., stopped making progress. Different algorithms use
  different conditions to decide that %
  a round is considered dead. 

\item A server \defn{learns} a value upon either receiving \defn{votes}
  from a quorum of servers for that value in some round or a message from
  some other server informing it about the value learned by that server. 

\item If the server is maintaining a state machine, it then \defn{executes}
  the learned value on the state machine. 

\item The \defn{result} of executing a value is sent to the client as a
  \defn{response} to the value it initially \defn{requested}.
\end{itemize}

\section{Language}

To describe properties, especially liveness properties, of distributed
systems, we introduce a precise high-level language. It supports direct use
of sets and predicate logic to capture system state and state properties, and
extends them with time and temporal expressions to capture properties
related to time.

\mypar{System state}
The following set and predicate model processes and process failures in the
system:
\begin{itemize}

\item \co{servers} --- the set of all servers in the system.  
This set may be dynamically updated, because servers can be added or removed.

\item \co{clients} --- the set of all clients in the system. 
Clients send values to servers for the servers to form consensus on.

\item %
\co{\p{p}.nf} --- true iff process \p{p} is non-faulty. %

\end{itemize}
To model messages, there are two built-in history variables for each
process \p{p}:
\begin{itemize}

\item \co{\p{p}.sent} --- the sequence of all messages sent by \p{p},
  optionally including the receiver process with each message.

\item \co{\p{p}.received} --- the sequence of all messages received by
  \p{p}, optionally including the sender process with each message.

\end{itemize}
For ease of reading, we use \co{\p{p}.sent \p{m} to \p{p2}} to mean \co{(\p{m},\p{p2}) in \p{p}.sent}, and use \co{\p{p}.received} \co{\p{m} from \p{p2}} to mean \co{(\p{m},\p{p2}) in \p{p}.received}. The clauses \co{to \p{p2}} and \co{from \p{p2}} are optional.

For consensus algorithms, the following variables and predicates are used:
\begin{itemize}
\item \co{quorums} --- the set of quorums of servers used by the algorithm.

\item \co{values} --- the set of proposed
  values from which a value or a sequence of values can be agreed on.

\item \co{\p{p}.is\_primary} --- whether \p{p} is currently the primary server.

\item \co{rounds} --- the total-ordered set of round identifiers.  Each uniquely determines a server and its current attempt at being the unique primary among all servers.

\end{itemize}\Vex{-1}

\mypar{Set expressions and quantifications}
\notes{
We use patterns for matching messages:
\begin{itemize}

\item A pattern can be used to match a message, in \co{sent} and
  \co{received}, and by a \co{receive} definition.  A constant value, such as
  \co{'response'}, or a previously bound variable, indicated with prefix
  \co{=}, in the pattern must match the corresponding components of the
  message.  An underscore matches anything.  Previously unbound variables
  in the pattern are bound to the corresponding components in the matched
  message.

  For example, \co{received ('response',=n,\_) from a} matches
  every triple in \co{received} whose first two components are
  \co{'response'} and the value of \co{n}, and binds \co{a} to the sender.
\end{itemize}
}
We use queries for expressing computations at a high level.
A query can be a comprehension, aggregation, or %
quantification over sets or sequences.
\begin{itemize}

\item A comprehension, also called a set former,
\arxiv{\co{\{\p{e}:~\p{v\sb{1}} in \p{s\sb{1}}, ..., \p{v\sb{k}} in \p{s\sb{k}}, \p{cond}\}},}%
\forlics{
\begin{code}
\quad \{\p{e}:~\p{v\sb{1}} in \p{s\sb{1}}, ..., \p{v\sb{k}} in \p{s\sb{k}}, \p{cond}\},
\end{code}
}
returns the set of values of \co{\p{e}} for all combinations of
  values of variables that satisfy all \co{\p{v_i} in \p{s\sb{i}}} clauses
  and condition \co{\p{cond}}.
  
\item An aggregation, \co{\p{agg} \p{s}}, where \co{\p{agg}} is an
  aggregation operator such as \co{count} or \co{max}, returns the value of
  applying \co{\p{agg}} to the set value of \co{s}.
  
\item A universal quantification,
\arxiv{\co{\EACH \p{v\sb{1}} in \p{s\sb{1}}, ...,
    \p{v\sb{k}} in \p{s\sb{k}} \HAS \p{cond}},}%
\forlics{
\begin{code}
\quad \EACH \p{v\sb{1}} in \p{s\sb{1}}, ...,
    \p{v\sb{k}} in \p{s\sb{k}} \HAS \p{cond},
\end{code}
}
returns true iff for each
  combination of values of variables that satisfy all \co{\p{v\sb{i}} in
    \p{s\sb{i}}} clauses, \co{\p{cond}} holds.

\item An existential quantification,
\arxiv{\co{\SOME \p{v\sb{1}} in \p{s\sb{1}},
    ..., \p{v\sb{k}} in \p{s\sb{k}} \HAS \p{cond}},}%
\forlics{
\begin{code}
\quad \SOME \p{v\sb{1}} in \p{s\sb{1}},
    ..., \p{v\sb{k}} in \p{s\sb{k}} \HAS \p{cond},
\end{code}
}
returns true iff for
  some combination of values of variables that satisfy all \co{\p{v\sb{i}}
    in \p{s\sb{i}}} clauses, \co{\p{cond}} holds.
When the query returns true, variables \p{v_1}, \m{\ldots}, \p{v_k} are bound to
  a combination of satisfying values, called a witness.
  
\end{itemize}\Vex{-1}

\mypar{Time and time scope}
There is a built-in notion of time, a number that ranges 
from the start of time, \co{0}, to unbounded time later, \co{inf}.  We omit the unit for time because the discussion is general and any unit can be used.

A time interval denotes all times between two times.
Time intervals are represented using common notation for intervals,
For example, \co{(t1,t2]} is the time interval from \co{t1} and \co{t2} excluding \co{t1} but including \co{t2}, and \co{[0,inf)} is the entire range of time.
A duration is the difference between two times.

To specify temporal properties, each expression has a time parameter specifying when the value of the expression is taken.
\begin{itemize}
\item \co{\p{e} \AT \p{t}} --- the value of expression \p{e} at time \p{t}.
This includes %
the case that \p{e} is a Boolean-valued condition. %

\item \co{\p{e} \AT{} .} --- for when the time parameter is implicit, 
where ``\co{.}'' (dot) denotes the implicit time in the scope of \p{e}.
For an outermost expression, the implicit time is \co{0}.

\end{itemize}
In any scope, the time parameter distributes to subexpressions, that is,
\co{\AT} is distributive.  For example, \co{(x + y) \AT t} \m{=} \co{(x \AT
  t) + (y \AT t)}.

\mypar{Temporal expressions}
There are two most important temporal operations:
\co{\ALW}, read as ``always'', and \co{\EVT}, read as ``eventually'',
defined as follows, where \p{t} is a fresh variable not used in the scope
of the expression being defined.
\begin{itemize}
\item \co{\ALW \p{cond} \m{=} \EACH \p{t} in [., inf) \HAS (\p{cond} \AT
    \p{t})}.  That is, \p{cond} holds at each time starting from the time
  in scope.

\item \co{\EVT \p{cond} \m{=} \SOME \p{t} in [., inf) \HAS (\p{cond} \AT
    \p{t})}.  That is, \p{cond} holds at some time starting from the time
  in scope.
\end{itemize}
The implicit time can be made explicit by following the definitions.  
For example, for an expression at the outermost scope, we have the following:
\begin{code}
    \ALW \EVT \p{cond} \m{=} (\ALW (\EVT \p{cond})) \AT 0
                  \m{=} \EACH t in [0,inf) \HAS ((\EVT \p{cond}) \AT t)
                  \m{=} \EACH t in [0,inf) \HAS ((\SOME t2 in [t, inf) \HAS \forlics{
                                             }(\p{cond} \AT t2)) \AT t)
\end{code}
For convenience, we define operator \co{\DURING}, where again \p{t} is a fresh variable:
\begin{itemize}
\item \co{\p{cond} \DURING \p{intvl} 
    \m{=} \EACH \p{t} in \p{intvl}:~\p{cond} \AT \p{t}}
\end{itemize}
We also define operators \co{\LASTS} and \co{\AFTER}:
\begin{itemize}

\item \co{\p{cond} \LASTS \p{dur} \m{=} \p{cond} \DURING [., .+\p{dur}]}

\item \co{\p{cond} \AFTER \p{dur} \m{=} \p{cond} \DURING (.+\p{dur}, inf)}

\end{itemize}\Vex{-1}

\mypar{Shorthand for non-faulty processes}
Finally, we define a shorthand for stating whether a set \p{ps} of
processes are all non-faulty, where \p{p} is a fresh variable:
\begin{code}
    \p{ps} nf \m{=} \EACH \p{p} in \p{ps} \HAS \p{p}.nf
\end{code}

\notes{
define, not good
unchanged(var) \DURING dur = \EACH t in [., .+dur] \HAS var \AT t = var \AT .

\begin{code}
    \p{cond} \AFTER \p{d} \m{=} \p{cond} \DURING (.+\p{d}, inf)
\end{code}
not yet used for consensus:
\begin{code}
    \p{cond} since \p{t} \m{=} \p{cond} \DURING [\p{t}, .]
    \p{cond} until \p{t} \m{=} \p{cond} \DURING [., \p{t})
\end{code}
after, since, until can be ambiguous: after a time or a duration. use time and dur.

yet no use of temporal logic operators: next p, p until q

operator \co{leadsto}, read as ``leads to'':
\begin{code}
    \p{cond} leadsto \p{cond2} \m{=} \p{cond} implies (\EVT \p{cond2}) 
\end{code}
operator \co{infoft}, read as ``infinitely often'', meaning also ``continually'', vs ``continuously'' for \co{\ALW}:
\begin{code}
    infoft \p{cond} \m{=} \ALW \EVT \p{cond}
\end{code}

not used: action and trigger:
set timer, timeout
} %

\section{Liveness specification and liveness hierarchy}
\label{sec-live-hierarchy}

We formally specify the liveness properties for all the algorithms and
variants listed in Table~\ref{tab-algo} as stated by their authors,
including assumptions about the communication links and about the servers,
which we call link assumptions and server assumptions, respectively.
We create a total order of link assumptions, a partial order of server
assumptions, and a partial order of liveness properties.  Elements in these
three hierarchies can be straightforwardly combined to form a single
liveness hierarchy.

Table~\ref{tab-link} shows the total order of link assumptions, in
increasing strength from top to bottom.
Figure~\ref{fig-live-hierarchy} shows the hierarchies for server
assumptions (on the left) and liveness assertions (on the right).
A solid arrow from node $A$ to node $B$ denotes that $A$ implies $B$, that
is, if $A$ is true, then $B$ must be true; we say that $A$ is
\defn{stronger} than $B$.
A dashed arrow (only from \Resp to \EachExec on the right) denotes that if
$A$ is true, then $B$ may be true.
Note that weaker assumptions are better, and stronger assertions are
better.  To relate to
Lamport's notion of weak and strong fairness~\cite{lamport1994temporal},
note that the assumptions made by weak fairness are stronger than the
assumptions made by strong fairness.

The total order for link assumptions is proved precisely in
Theorem~\ref{thm-link}.  The partial orders for server assumptions and
liveness assertions are argued informally but can be proved precisely in a
similar fashion as for link assumptions.

\subsection{Liveness assumptions}

\Vex{-2}
\mypar{Link assumptions} 
Each algorithm makes some assumptions about the communication links used,
forming three different kinds of link assumptions, which are used to
separate the algorithms.
\begin{enumerate}

\item \Raw{} --- messages sent may or may not be received.
  This includes algorithms and variants that do not discuss liveness
  propositions, or do not precisely describe the assumptions made about
  links in their liveness discussion.  For example, VR-Revisit only states
  that some processes ``are able to
  communicate''~\cite[Section~8]{liskov2012viewstamped}.

\item \Fair{} --- all links between servers are \defn{fair}, that is ``if a
  correct process repeatedly sends a message to another correct process, at
  least one copy is eventually delivered''~\cite{van2015paxos}.
We also include algorithms and variants that assume no communication failures in this category under the assumption that copies of a same message are identical, meaning that retransmitting a message does not change the set of sent messages. With this view, fair links assumptions states that every %
message in the set of sent messages is eventually delivered---which
coincides with no communication failure. %

\item \Sure{} --- the time between a message being sent and the message
  being received has a known upper bound. Note that this implies that every
  sent message is received. This differs from fair links by having a known
  upper bound.

\end{enumerate}

Figure~\ref{fig-link-formal} shows the three kinds of link assumptions
specified formally, where the known bound for \Sure is captured by the
explicit parameter \m{D}.  It is easy to see precisely that the difference
between \Raw and \Fair is the outermost existential vs.\ universal
quantification, and the difference between \Fair and \Sure is the eventual
condition vs.\ the known duration.\m{\!}
\begin{figure*}
\begin{lstlisting}
Raw     = some p1.sent m to p2 has evt p2.received m from p1

Fair    = each p1.sent m to p2 has evt p2.received m from p1

Sure(D) = each p1.sent m to p2 has (p2.received m from p1 after D)
\end{lstlisting}\podc{\Vex{-2}}\disc{\Vex{-2}}\arxiv{\Vex{-1}}
\notes{alternatives/extended for proofs:
Raw     = some p1.sent m to p2 at t has 
                 (some d in [0, inf) has p2.received m from p1 at t+d)
                 
Fair    = each p1.sent m to p2 at t has 
                 (some d in [0, inf) has p2.received m from p1 at t+d)
                 (some t2 in [t, inf) has p2.received m from p1 at t2)

Sure(D) = each p1.sent m to p2 at t has 
                 (some d in [0, D] has p2.received m from p1 at t+d)
                 (some t2 in [t, t+D] has p2.received m from p1 at t2)
}
\notes{
todo: 
d could start at 0 not after 0

possibly add direct temporal actions later:
def fair\_links():
    eventually(always(
        each(p1 in processes(), p2 in processes(), has=
            implies(p1.alive and p2.alive,\
                    leads_to(infinitely_often(p1.send(m, to= p2)),
                             p2.received(m, from_= p1)))
                             #m from p1 in p2.received))) 
                             #change all to ideal later
    )))
}
\caption{Link assumptions expressed precisely.\podc{\Vex{-1}}\disc{\Vex{-1}}}
\label{fig-link-formal}
\end{figure*}

Table~\ref{tab-link} shows the algorithms and variants in
Table~\ref{tab-algo} grouped into three categories based on the link
assumptions they stated.  We see that almost half of the 31 algorithms and
variants fall under \Raw.
\providecommand{\tablinkwidth}{\forlics{34}\podc{92}\disc{84}\arxiv{79}ex}
\begin{table}
\small
\centering
\begin{tabular}{@{}c@{\,}||@{~}p{\tablinkwidth}@{~}|@{\,}c@{}}
    Assumption & \multicolumn{1}{c|@{\,}}{Algorithms and Variants} & Count
    \\\hline\hline
    \Raw
    & 
    VS-ISIS, VS-ISIS2, Derecho, VR, VR-Revisit, Paxos-Basic, Paxos-Vertical, Paxos-PVS, Chubby, Zab-FLE, Raft, Paxos-TLA, Raft-Verdi, Paxos-Decon, Paxos-High 
    & 15
    \\\hline
    \Fair
    &
    EVS, Paxos-VS, Paxos-Fast, CT, ACT, Paxos-SB, Mencius, EPaxos, Paxos-Complex, Paxos-EPR 
    & 10
    \\\hline
    \Sure
    & 
    Paxos-Synod, Paxos-Time, Chubby-Live, Zab, IronRSL, LastVoting-PSync
    & 6
    \\\hline
\end{tabular}
\podc{\Vex{-2}}\disc{\Vex{-2}}\arxiv{\Vex{-1}}
\caption{Assumptions about links used in algorithms and variants in Table~\ref{tab-algo}.\disc{\Vex{-4}}}
\label{tab-link}
\end{table}

Theorem~\ref{thm-link} shows that the three kinds of link assumptions form
a total order.  It is equally easy to show that, for any duration \m{D1}
shorter than or equal to \m{D2}, \Sure(\m{D1}) implies \Sure\nolinebreak(\m{D2}).
\begin{theorem}
  \label{thm-link}
  \Fair implies \Raw, and \Sure implies \Fair.
\end{theorem}\disc{\Vex{-4}}
\begin{proof}
  By definitions of \Raw, \Fair, and \Sure, and definitions of \co{\EVT}
  and \co{\AFTER}, we have the following:
\begin{code}
\arxiv{  }Raw    \arxiv{
  } = \SOME p1.sent m to p2 \AT t \HAS \SOME d in [0,inf)\forlics{
              } \HAS (p2.received m from p1 \AT t+d)
\arxiv{  }Fair   \arxiv{
  } = \EACH p1.sent m to p2 \AT t \HAS \SOME d in [0,inf)\forlics{
              } \HAS (p2.received m from p1 \AT t+d)
\arxiv{  }Sure(D) \arxiv{
  } = \EACH p1.sent m to p2 \AT t \HAS \SOME d in [0,D]  \forlics{
              } \HAS (p2.received m from p1 \AT t+d)
\end{code}
When any message is sent, top-level universal quantification (\co{\EACH})
in \Fair implies the existential quantification (\co{\SOME}) in \Raw, thus
\Fair implies \Raw.  
Some \co{d} in \co{[0,D]} in \Sure implies some \co{d} in \co{[0,inf)} in
\Fair, thus \Sure implies \Fair.
\end{proof}

\mypar{Server assumptions}
The server assumptions form a 7-element partial order, as described
below.  Figure~\ref{fig-server-formal} shows formal specifications of the server
assumptions.  Figure~\ref{fig-live-hierarchy}-left shows
the hierarchy they form.

\begin{enumerate}

\item \AlwQ{} --- eventually, there is always some quorum of servers that
  is non-faulty.  Note that no fixed quorum is non-faulty. This is the most
  basic server assumption because in quorum-based consensus algorithms, the
  system is said to be \defn{down} if no quorum is non-faulty. In other
  words, this assumption states that eventually, the system is always
  \defn{up}.  This is called Strong-L1 in~\cite{kirsch2008paxos}. We
  include EPaxos in this category.

\item \QAlw{} --- eventually, there is some quorum of servers that is always
  non-faulty.
  This is called Weak-L1 in~\cite{kirsch2008paxos} and is used to argue
  liveness of Paxos-SB.  This assumption is stronger than \AlwQ because it
  assumes a \textit{fixed} non-faulty quorum. We also include EVS, CT, ACT,
  and Paxos-SB in this category.
\notes{
  Note that the liveness assumptions of CT and ACT state that eventually, a
  majority of servers is \textit{good}, where good is defined differently
  for each. Prima facie, this looks to be \AlwQ, however, we categorize
  them in \QAlw due to lemmas~\ref{lem-alwq-qalw-crash}
  and~\ref{lem-alwq-qalw-crash-recovery} from
  Section~\ref{sec-live-assm-lemmas}, respectively.  }

\item \PAlwQ{} --- eventually, there is a server P such that, always, P is
  non-faulty and is the primary and there is some quorum Q of servers that
  is non-faulty.
  \PAlwQ is stronger than \AlwQ because it assumes a fixed primary that is
  eventually always non-faulty.
  If we assume that primary is always part of the non-faulty quorum, then
  \PAlwQ can be interpreted as that after some time, it always holds that
  some quorum containing the primary is non-faulty. We include Paxos-VS,
  Derecho and VR-Revisit in this category.

\item \PQAlw{} --- eventually, there is a server P and a quorum Q of
  servers, such that always, P is non-faulty and is the primary and Q is
  non-faulty.
  This assumption is stronger than \QAlw because it assumes a fixed primary
  and is stronger than \PAlwQ because it assumes a fixed non-faulty quorum.
  Zab's liveness is proved using this. Paxos-Fast, IronRSL, and Paxos-EPR
  use a slight generalization of this assumption which substitutes a single
  non-faulty quorum with a set of non-faulty quorums.

\item \Alw{} --- eventually, all servers are always non-faulty, i.e.,
  non-faulty forever from that time on. This is obviously the strongest
  assumption.  This assumption is used while checking liveness of
  Chubby-live. The authors begin by testing for only safety while their
  failure injector injects random failures into the system. After some
  time, the failure injector is stopped and the system is given time to
  fully recover. Once recovered, they do not inject any more failures and
  check for liveness violations.
The purpose of their liveness test is to check that the system does not deadlock after a sequence of failures.
ACT gives liveness performance guarantees assuming this.
\end{enumerate}

\begin{figure*}
\begin{lstlisting}
Alw-Q     = evt alw some q in quorums has q nf

Q-Alw     = evt some q in quorums has alw q nf

P-Alw-Q   = evt some p in servers has
                alw (p.nf and p.is_primary and some q in quorums has q nf)

PQ-Alw    = evt some p in servers, q in quorums has
                alw (p.nf and p.is_primary and q nf)

Alw       = evt alw servers nf

PQ-Dur(D) = evt some p in servers, q in quorums has 
                ((p.nf and p.is_primary and q nf) lasts D)
                 
PQ-Extra-Dur(D1,D2) =
some t in [0,inf) has (
(each t2 in [t,t+D1+D2] has (servers at t2 = servers at t)) and
(some p in servers has ((p.nf and p.is_primary) during [t+D1,t+D1+D2])) and
(some q in quorums has (q nf during [t,t+D1+D2])) )
\end{lstlisting}
\podc{\Vex{-2}}\disc{\Vex{-2}}\arxiv{\Vex{-1}}
\notes{for old name, with alternatives/extended for proofs:
Evt-Alw-Q

Q-Evt-Alw   = \SOME q in quorums \HAS \EVT \ALW q nf

P-Evt-Alw-Q = \SOME p in servers \HAS \EVT \ALW
                (p.nf and p.is_primary and \SOME q in quorums \HAS q nf)
            = \EVT \SOME p in servers \HAS \ALW \SOME q in quorums \HAS
                (p.nf and p.is_primary and q nf)
            = \EVT \ALW (primary.nf and \SOME q in quorums \HAS q nf)?

PQ-Evt-Alw  = \SOME p in servers, q in quorums \HAS \EVT \ALW
                (p.nf and p.is_primary and q nf)
            = \SOME q in quorums \HAS \EVT \ALW (primary.nf and q nf)?
            = \EVT \SOME q in quorums \HAS \ALW (primary.nf and q nf)?

Evt-Alw     = \EVT each p in servers \HAS \ALW p.nf

PQ-Dur (D)
= \EVT (primary.nf \LASTS D and \SOME q in quorums \HAS ((q nf) \LASTS D)))
= \SOME t in [0, inf) \HAS
    ((primary.nf and \SOME q in quorums \HAS (q nf)) \DURING [t, t+D])

PQ-Extra-Dur (T, D)
= \SOME ... \HAS
    ... and
    (primary.nf \DURING [t+T, inf)) and # \AFTER T if \EVT, not \SOME t 
    ...

def PQ-Extra-Dur():
    eventually(
        during(unchanged(processes), duration= T1 + T2) and\
        after(leader != null, T1) and\
        during(\SOME(q in Quorums, \HAS= \EACH(p in q, \HAS= p.is_nf)),
            from_= T1, for_= T2)
    )
}
\caption{Server assumptions expressed precisely.\disc{\Vex{-1}}}%
\label{fig-server-formal}
\end{figure*}

The next two assumptions differ from the ones described above in the sense
that these two assume that eventually some set of processes remains
non-faulty for some duration of time, not eventually always.

\begin{enumerate}
\setcounter{enumi}{5}
\item \PQDur{} --- eventually, there is a server P and a quorum Q such that
  P and Q are non-faulty and P is the primary, for some (sufficiently long)
  duration of time. \PQDur is the equivalent of \PQAlw but with time
  analysis. Paxos-Time and LastVoting-PSync use this assumption.

\item \PQExtraDur.
Lamport states the following about liveness of Paxos-Synod in~\cite{lamport1998part}:

\begin{displayquote}
  \textit{Presidential selection requirement}:~
  If no one entered or left the Chamber, then after $T$ minutes exactly one
  priest in the Chamber would consider himself to be the president.

  \textit{Liveness property}:~
  If the presidential selection requirement were met, then the complete
  protocol would have the property that if a majority set of priests were
  in the chamber and no one entered or left the Chamber for $T + 99$
  minutes, then at the end of that period every priest in the Chamber would
  have a decree written in his ledger.
\end{displayquote}

Using $D1$ and $D2$ instead of $T$ and $99$, respectively, we can summarize
the assumptions as having a duration of time lasting $D1 + D2$ time units
such that (1) the set of non-faulty servers does not change in the duration,
(2) after $D1$ time units, primary is selected,
and (3) at least a fixed quorum of servers is non-faulty during the entire
duration - this is because ``no one enters or leaves the chamber'' and ``a
majority of priests were in the chamber ... for $D1 + D2$ minutes''.

\PQExtraDur is weaker than \Alw because it allows some servers to be faulty
in the pertinent duration.
\PQExtraDur is stronger than \PQDur, but is not comparable to \PQAlw---it
is weaker in using limited durations instead of \co{\ALW}, and it is
stronger in requiring that no server fails during $D1 + D2$ whereas \PQAlw
allows for servers to fail as long as P and Q are non-faulty

\end{enumerate}

\notes{
\subsection{Liveness assumption lemmas}
\label{sec-live-assm-lemmas}

\begin{lemma}\label{lem-alwq-qalw-crash}
\AlwQ implies \QAlw in crash failure model.
\end{lemma}
\begin{proof}
We prove this by contradiction. Because \AlwQ does not imply \QAlw, there exists a system and an execution in which, \AlwQ holds but \QAlw does not. An execution that satisfies this assumption is: there exist quorums $Q1$ and $Q2$ and times $t1$ and $t2$, $t1 < t2$, such that at time $t1$, $Q1$ is non-faulty but $Q2$ is faulty and at time $t2$, $Q2$ is non-faulty but $Q1$ is faulty. Because $Q1 \neq Q2$, this satisfies \AlwQ but not \QAlw.

Because $Q2$ is faulty at time $t1$, some server $q \in Q2$ is faulty at $t1$. From crash failure model, we conclude that $q$ must be faulty at $t2$ because $t2 > t1$. This contradicts our assumption that $Q2$ is non-faulty at $t2$.
\end{proof}

\notes{
\begin{lemma}\label{lem-alwq-qalw-crash-recovery}
\AlwQ implies \QAlw if eventually, it is always the case that a quorum exists such that every server in the quorum is either always non-faulty or eventually always non-faulty.
\end{lemma}
\begin{proof}
We prove this by contradiction as well. Like the proof of~\ref{lem-alwq-qalw-crash}, we again assume an execution in which there exist quorums $Q1$ and $Q2$ and times $t1$ and $t2$, $t1 < t2$, such that at time $t1$, $Q1$ is non-faulty but $Q2$ is faulty and at time $t2$, $Q2$ is non-faulty but $Q1$ is faulty. Because $Q1 \neq Q2$, this satisfies \AlwQ but not \QAlw.

Because $Q2$ is faulty at time $t1$, some server $q \in Q2$ is faulty at $t1$.
\end{proof}
}
}

\subsection{Liveness assertions}

The liveness assertions form a 6-element partial order, as described
below; for brevity, we omit slots from the description below.
Figure~\ref{fig-live-formal} shows formal specifications of the liveness assertions; Figure~\ref{fig-live-formal-multi} shows the
corresponding assertions with slots for multi-value consensus.
Figure~\ref{fig-live-hierarchy}-right shows the hierarchy these liveness
assertions form.

\notes{
Assertions with nf: at start of second line: add "q nf and" for 3 Each-'s
and "p.nf and" for 2 Some-'s. 
}
\begin{figure*}
\begin{lstlisting}
Each-Vote  = evt some r in rounds, q in quorums, v in values has
                 each p in q has p.voted (r,v)

Some-Learn = evt some p in servers, v in values has p.learned (v)

Each-Learn = evt some q in quorums, v in values has
                 each p in q has p.learned (v)

Some-Exec  = evt some p in servers, v in values has p.executed (v)

Each-Exec  = evt some q in quorums, v in values has
                 each p in q has p.executed (v)

Resp       = evt each c in clients has some v in values has
                 c.received ('resp',v)
\end{lstlisting}
\podc{\Vex{-2}}\disc{\Vex{-2}}\arxiv{\Vex{-1}}
\notes{
Each-Vote = \EVT \EACH p in quorums \HAS \SOME v in values \HAS q voted (v)

Some-Learn  = \EVT \SOME p in servers \HAS \SOME v in values \HAS p.decided (v)

Each-Learn  = \EVT \EACH q in quorums \HAS \SOME v in values \HAS q decided (v)

Some-Exec = \EVT \SOME p in servers \HAS \SOME v in values \HAS p.executed (v)

Each-Exec = \EVT \EACH q in quorums \HAS \SOME v in values \HAS q executed (v)

Resp      = \EACH c in clients \HAS 
              \EACH c.sent ('request', v) \HAS \EVT c.received ('resp', v)
}
\caption{Liveness assertions expressed precisely for single-value consensus. 
For a server \co{p}, round \co{r}, and value \co{v},
\co{~p.voted (r,v)} is true iff \co{p} voted for value \co{v} in round \co{r},
\co{~p.learned (v)} is true iff \co{p} learned that value \co{v} was chosen,
\co{~p.executed (v)} is true iff \co{p} has executed \co{v} in its local copy of the replicated state machine.
For a client \co{c} and value \co{v}, \co{~c.received (\msg{'resp'},v)} is true iff \co{c} received value \co{v} in response.\Vex{-1}
}
\label{fig-live-formal}
\end{figure*}\podc{\Vex{-2}}\disc{\Vex{-2}}

\notes{
}

\begin{enumerate}
\item \EachVote{}
  --- eventually in some round, each server of some quorum sends a vote for the same value in that round. %
  \EachVote is crucial for Paxos and its variants because Paxos' safety relies on the invariant that if a quorum votes on the same value in some round, then only that value can be proposed in subsequent rounds.
  \EachVote property is formally verified for Paxos-EPR, as well as single-value Paxos and Stoppable Paxos~\cite{malkhi2008stoppable} in the proof system Ivy~\cite{padon2016ivy}. %
    
\item \SomeLearn{} --- eventually, some non-faulty server learns a value.
  This is stronger than \EachVote, because the sent votes might not be received by any server.
  In Paxos-Fast, it is informally proven that Paxos-Basic %
  and Paxos-Fast satisfy this property.
  We consider this, not \EachVote, as the most basic property that any live
  implementation of any consensus algorithm is expected to guarantee.
  Figure~\ref{fig-RaftEachVote} shows that Raft admits behaviors of
  infinite length in which \EachVote is eventually always satisfied but
  \SomeLearn is never satisfied.
  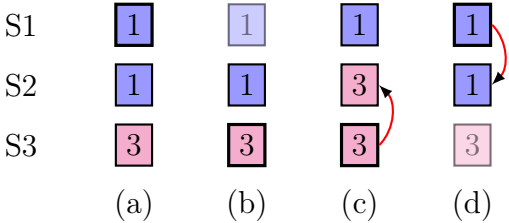
\begin{figure}
    \centering
    \begin{tikzpicture}[node distance=8mm,>=stealth',auto]
    
      \tikzstyle{logSlot1}=[rectangle,thick,draw=black,fill=blue!40,minimum size=4mm]
      \tikzstyle{logSlot2}=[logSlot1,fill=magenta!40]
      \tikzstyle{primary}=[very thick]
      \tikzstyle{logRepl}=[thick,-latex,bend right,draw=red]
      			  
      \tikzstyle{every label}=[red]
    
      \begin{scope}
        \node (s1) {S1};
        \node (s2) [below of=s1] {S2};
        \node (s3) [below of=s2] {S3};
      \end{scope}
    
      \begin{scope}[xshift=1.5cm]
        \node [logSlot1] (t1s1) [primary]    {1};
        \node [logSlot1] (t1s2) [below of=t1s1] {1};
        \node [logSlot2] (t1s3) [below of=t1s2] {3};
        \node (t1cap) [below of=t1s3] {(a)};
      \end{scope}
      
      \begin{scope}[xshift=3.0cm]
        \node [logSlot1] (t2s1) [opacity=0.5] {1};
        \node [logSlot1] (t2s2) [below of=t2s1] {1};
        \node [logSlot2] (t2s3) [below of=t2s2,primary] {3};
        \node (t2cap) [below of=t2s3] {(b)};
      \end{scope}
    
      \begin{scope}[xshift=4.5cm]
        \node [logSlot1] (t3s1) {1};
        \node [logSlot2] (t3s2) [below of=t3s1] {3};
        \node [logSlot2] (t3s3) [below of=t3s2,primary] {3};
        \node (t3cap) [below of=t3s3] {(c)};
        
        \draw [logRepl] (t3s3.east) to [bend right=45] (t3s2.east);
      \end{scope}
    
      \begin{scope}[xshift=6.0cm]
        \node [logSlot1] (t4s1) [primary] {1};
        \node [logSlot1] (t4s2) [below of=t4s1] {1};
        \node [logSlot2] (t4s3) [below of=t4s2, opacity=0.5] {3};
        \node (t4cap) [below of=t4s3] {(d)};
        
        \draw [logRepl] (t4s1.east) to [bend left=45] (t4s2.east);
      \end{scope}
    
    \end{tikzpicture}\podc{\Vex{-2}}
    \caption{Demonstration of Raft satisfying \EachVote but not
      \SomeLearn. The system involves three servers, S1, S2, and S3. The
      primary is marked with a thick border. No value has been committed
      yet. A faded server is currently failed. (a) The initial state of the
      server logs. S1 has proposed value 1 which S2 has received and
      written in its log thus satisfying \EachVote for value 1. Assume that
      S1 fails before realizing the majority, thus \SomeLearn does not
      satisfy. (b) S3 becomes the new primary. (c) S3 enforces its log onto
      S2 replacing S1's proposed value. Now \EachVote is satisfied for
      value 3. (d) Before \SomeLearn could be satisfied, S3 fails and S1
      becomes the new leader. It enforces its log onto S2. Thus, \EachVote
      is satisfied for value 1 coming back to configuration (a).}
    \label{fig-RaftEachVote}
  \end{figure}
    
\item \EachLearn{} --- eventually, each server of some quorum learns the
  same value.  Obviously, \EachLearn is stronger than \SomeLearn. This
  property is informally proven for Paxos-Synod and formally proven for
  LastVoting-Psync.
    
\item \SomeExec{} --- eventually, some non-faulty server executes a value.
  This is stronger than \SomeLearn because a value has to be learned before
  it can be executed. \SomeExec and \EachLearn do not compare because
  execution details differ between algorithms. While some would allow
  execution immediately after learning, others might require replication to
  occur before execution. It is claimed that Paxos-SB and EPaxos satisfy
  this property, and formally verified in Dafny that IronRSL satisfies this
  property. %
    
\item \EachExec{} --- eventually, each server of some quorum executes the same
  value. Obviously, this is stronger than \SomeExec and \EachLearn. It is
  informally proven that EVS satisfies this property. In Paxos-Time, IOA
  specifications of single-value Paxos and Paxos are informally proven to
  satisfy this property. A proof sketch is given for Zab satisfying this
  property.
    
\item \Resp{} --- eventually, each client request is responded to. This
  property is stronger than \SomeExec because client expects the result of
  executing the requested value. \Resp may or may not be stronger than
  \EachExec. For example, normal case execution of Paxos-VS satisfies
  \EachExec before \Resp, that is, \Resp implies \EachExec for
  Paxos-VS. Meanwhile, normal case execution of VR-Revisit satisfies
  \SomeExec and then \Resp but not \EachExec, that is, \Resp implies
  \SomeExec but not \EachExec for VR-Revisit. Hence, the dashed arc from
  \Resp to \EachExec.
  \notes{Birman et al~\cite{birman2010virtually} claim that Paxos-VS %
    satisfies this property. Liskov and Cowling~\cite{liskov2012viewstamped} informally argue that VR-Revisit %
    satisfies this. Chandra et al~\cite{chandra2007paxos} test that Chubby-Live %
    ensures this. van Renesse and Altinbuken~\cite{van2015paxos} claim this property for their Paxos-Complex. Hawblitzel et al~\cite{...} give a machine-checked proof which verifies that their Iron...
    satisfies this property.
  }
\end{enumerate}

\begin{figure*}
\begin{lstlisting}
Each-Vote(n) = evt each s in 1..n has
                   some r in rounds, q in quorums, v in values has
                   each p in q has p.voted (r,s,v)

Some-Learn(n)= evt each s in 1..n has
                   some p in servers, v in values has p.learned (s,v)

Each-Learn(n)= evt each s in 1..n has
                   some q in quorums, v in values has
                   each p in q has p.learned (s,v)

Some-Exec(n) = evt each s in 1..n has
                   some p in servers, v in values has p.executed (s,v)

Each-Exec(n) = evt each s in 1..n has
                   some q in quorums, v in values has
                   each p in q has p.executed (s,v)

Resp         = evt each c in client, v in values has
                   c.sent ('req', v) implies c.received ('resp', v, res(v))
\end{lstlisting}
\podc{\Vex{-2}}\disc{\Vex{-2}}\arxiv{\Vex{-1}}
\caption{
Liveness assertions expressed precisely for multi-value consensus.
For a server \co{p}, round \co{r}, slot \co{s}, and value \co{v},
\co{~p.voted (r,s,v)} is true iff \co{p} voted for value \co{v} in slot \co{s} in round \co{r},
\co{~p.learned (s,v)} is true iff \co{p} learned that value \co{v} was
chosen in slot \co{s},
\co{p.executed (s,v)} is true iff \co{p} has executed \co{v} in its local copy of the replicated state machine in slot \co{s}.
For a client \co{c} and value \co{v}, \co{~c.sent (\msg{'req'}, v)} is true
iff \co{c} requested \co{v} to be executed, 
\co{~c.received (\msg{'resp'}, v, res(v))} is true iff \co{c} received a
response that value \co{v} was executed with result \co{res(v)}.
}
\label{fig-live-formal-multi}
\end{figure*}

For multi-value consensus, the liveness assertions are specified precisely
in Figure~\ref{fig-live-formal-multi}. %
Compared with assertions for single-value consensus in
Figure~\ref{fig-live-formal}, the changes are:
\begin{itemize}
\item for each assertion except for \co{Resp}, add (1) parameter \co{n} for
  the assertion, (2) clause \co{each s in 1..n} on the first line, and (3)
  component \co{s} before \co{v} to the arguments of \co{sent} and
  \co{executed} on the last line.

\item for assertion \co{Resp}, add 
  (1) clause \co{v in values} on the first line,
  (2) the second line giving the condition for implication, and
  (3) the last component \co{res(v)} to the argument of \co{received} 
  on the last line. \co{res(v)} is the result of executing value \co{v}.
\end{itemize}

\notes{
original and others tried:
\begin{lstlisting}[linewidth=92ex]
Each-Vote(s)= evt some q in quorums, r in rounds, v in values has # param s
                  each p in q has p.sent ('vote',r,s,v)           # component s
Each-Vote   = evt each s in Nat has                               # each clause
                  some q in quorums, r in rounds, v in values has
                  (some p in q has p.sent ('vote',r,s,v) implies  # all
                   each p in q has p.sent ('vote',r,s,v))         # component s
Each-Vote   = evt some s has Each-Vote(s) and                     # start naive to derive
                  each s has sent (vote s) implies evt Each-Vote(s)
            = evt some s has sent (vote S) and
                  ... (same as above)
            = evt some sent('vote',r,s,v) and
                  each sent('vote',r,s,v) has 
                  evt some q in quorums has each p in q has p.sent('vote',r,s,v)
            = evt some r in rounds, s in Nat, v in values has .sent('vote',r,s,v) and
                  each r in rounds, s in Nat, v in values has 
                  (sent('vote',r,s,v) implies 
                   evt some q in quorums has each p in q has p.sent('vote',r,s,v))
            = ... # continued in the notes below, but abandoned as explained in the notes

Some-Learn(s) = evt some p in servers, v in values has    # parameter s
                  (p.nf and p.sent ('learn',s,v))      # component s

Each-Learn(s) = evt some q in quorums, v in values has    # parameter s
                  each p in q has p.sent ('learn',s,v) # component s
Each-Learn    = evt each s in Nat has                     # each clause
                  some q in quorums, v in values has
                  some p in q has p.sent ('learn',s,v) and # all
                  each p in q has p.sent ('learn',s,v) # component s

# below is still single value, and includes safety by having some v before each c 
Resp        = evt some v in values has each c in clients has
                  c.received ('resp', v)
\end{lstlisting}
} %
\notes{
continue from notes for derivation above:

\EachVote below doss not look right: \EACH sent vote gets a quorum to vote the same?
Better with that second eventually, but we are already in an \EVT context.

It is also not satisfied by Paxos-Complex.
It is also not implied by \SomeExec(n) or \EachExec(n).

Do the theorems that use \EachVote need to say anything about slots?
If not, just use the parameterized version.

Conclusion: Better/simpler just use the up-to-n version for all.

\begin{lstlisting}[linewidth=92ex]
Each-Vote   = evt some r in rounds, s in Nat, v in values, p in servers has
                  p.sent('vote',r,s,v) and
                  each r in rounds, s in Nat, v in values has 
                      some p in servers has p.sent ('vote',r,s,v) implies 
                      evt some q in quorums has each p in q has
                          p.sent ('vote',r,s,v)

(*\notes{
Some-Learn    = evt some r in rounds, s in Nat, v in values, p in servers has
                  p.sent('vote',r,s,v) and
                  each r in rounds, s in Nat, v in values has 
                      some p in servers has p.sent ('vote',r,s,v) implies 
                      evt some p in servers has p.sent ('learn',s,v) 
Wrong because decide does not have round.
so this enforces that every <r,s,v> triple that gets voted on gets decided. 
So if voted(r1, s, v1) and voted(r2, s, v2), then v1 has to be equal to v2, otherwise the algorithm is not live. For eg, Raft may not satisfy this.
So, this may not be too restrictive.
}*)
Some-Learn    = evt some r in rounds, s in Nat, v in values, p in servers has
                  p.sent('vote',r,s,v) and
                  each s in Nat has 
                      some p in servers has p.sent ('vote',_,s,_) implies 
                      evt some p in servers, v in values has p.sent ('learn',s,v)

Each-Learn    = evt some r in rounds, s in Nat, v in values, p in servers has
                  p.sent('vote',r,s,v) and
                  each s in Nat has 
                      some p in servers has p.sent ('vote',_,s,_) implies 
                      evt some q in quorums has each p in q has
                          p.sent ('learn',s,v)
\end{lstlisting}
}%

\newcommand{\mybullet}{\textbullet\hspace{0.25em}}
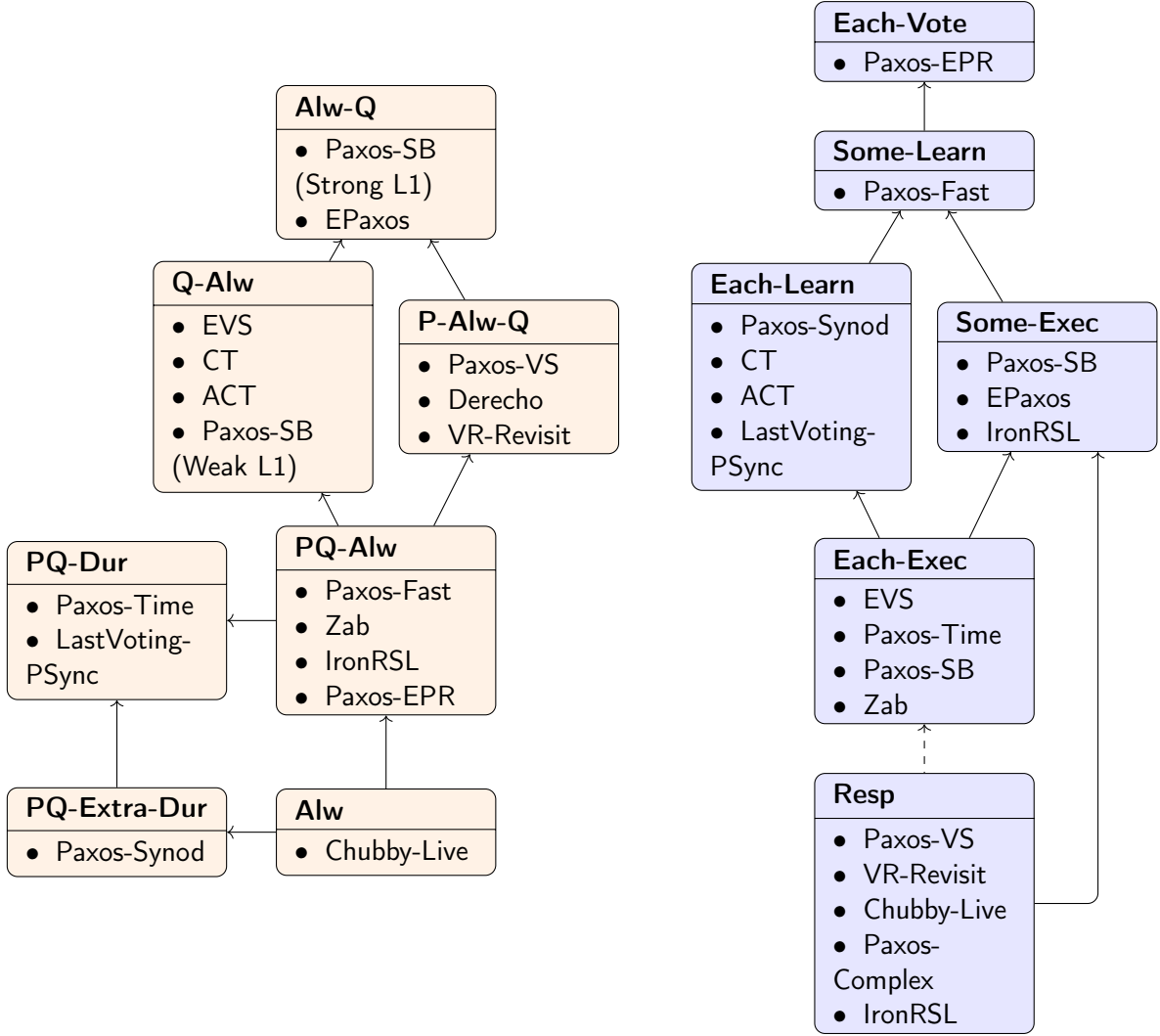
\begin{figure*}[htp]
\small\disc{\Vex{-4}}
\begin{center} 
\begin{tikzpicture}[node distance=1.5cm,
    every node/.style={fill=orange!10,
                        rectangle split,
                        rectangle split parts=2,
                        rounded corners,
                        draw=black,
                        minimum width={18ex},
                        text width={15ex},
                        minimum height=1cm,
                        font=\sffamily,
                        align=left},
    align=left,
    property/.style={fill=blue!10},
    invisible/.style={opacity=0},
    ]

  \node (Alw-Q) %
    {
    \textbf{Alw-Q}
    \nodepart{second}
    \mybullet{} Paxos-SB (Strong L1)\\
    \mybullet{} EPaxos
    };

\node [invisible] (dummy3) [below = of Alw-Q] {};
  
  \node (Q-Alw) [left = -8ex of dummy3]
    {
    \textbf{Q-Alw}
    \nodepart{second}
    \mybullet{} EVS\\
    \mybullet{} CT\\
    \mybullet{} ACT\\
    \mybullet{} Paxos-SB (Weak L1)
    };

  \node (P-Alw-Q) [right = -8ex of dummy3]
    {
    \textbf{P-Alw-Q}
    \nodepart{second}
    \mybullet{} Paxos-VS\\
    \mybullet{} Derecho\\
    \mybullet{} VR-Revisit\\
    };

  \node (PQ-Alw) [below = 10ex of dummy3]
    {
    \textbf{PQ-Alw}
    \nodepart{second}
    \mybullet{} Paxos-Fast\\
    \mybullet{} Zab\\
    \mybullet{} IronRSL\\
    \mybullet{} Paxos-EPR
    };

  \node (Alw) [below = 6ex of PQ-Alw]
    {
    \textbf{Alw}
    \nodepart{second}
    \mybullet{} Chubby-Live
    };

  \node (PQ-Extra-Dur) [left = 4ex of Alw]
    {
    \textbf{PQ-Extra-Dur}
    \nodepart{second}
    \mybullet{} Paxos-Synod
    };

  \node (PQ-Dur) [left = 4ex of PQ-Alw]
    {
    \textbf{PQ-Dur}
    \nodepart{second}
    \mybullet{} Paxos-Time\\
    \mybullet{} LastVoting-PSync
    };

  \node [invisible] (dummy2) %
    [right = 16ex of P-Alw-Q] {};
  
  \node [property] (P) %
    [above = 22ex of dummy2]
    {
    \textbf{Each-Vote}
    \nodepart{second}
    \mybullet{} Paxos-EPR
    };
  
  \node [property] (Some-Learn) [below = 4ex of P]
    {
    \textbf{Some-Learn}
    \nodepart{second}
    \mybullet{} Paxos-Fast
    };
    
    \node [property] (Each-Learn) [left = -8ex of dummy2]
    {
    \textbf{Each-Learn}
    \nodepart{second}
    \mybullet{} Paxos-Synod\\
    \mybullet{} CT\\
    \mybullet{} ACT\\
    \mybullet{} LastVoting-PSync
    };

   \node [property] (Some-Exec) [right = -8ex of dummy2]
    {
    \textbf{Some-Exec}
    \nodepart{second}
    \mybullet{} Paxos-SB\\
    \mybullet{} EPaxos\\
    \mybullet{} IronRSL
    };

  \node [property] (Each-Exec) [below = 11ex of dummy2]
    {
    \textbf{Each-Exec}
    \nodepart{second}
    \mybullet{} EVS\\
    \mybullet{} Paxos-Time\\
    \mybullet{} Paxos-SB\\
    \mybullet{} Zab
    };
  
  \node [property] (Resp) [below = 4ex of Each-Exec]
    {
    \textbf{Resp}
    \nodepart{second}
    \mybullet{} Paxos-VS\\
    \mybullet{} VR-Revisit\\
    \mybullet{} Chubby-Live\\
    \mybullet{} Paxos-Complex\\
    \mybullet{} IronRSL
    };

  \draw[<-] (PQ-Alw) -- (Alw);
  \draw[<-] (Q-Alw) -- (PQ-Alw);
  \draw[<-] (Alw-Q) -- (Q-Alw);
  \draw[<-] (PQ-Extra-Dur) -- (Alw);
  \draw[<-] (PQ-Dur) -- (PQ-Alw);
  \draw[<-] (PQ-Dur) -- (PQ-Extra-Dur);
  \draw[<-] (P-Alw-Q) -- (PQ-Alw);
  \draw[<-] (Alw-Q) -- (P-Alw-Q);

  \draw[<-] (P) -- (Some-Learn);
  \draw[<-] (Some-Learn) -- (Some-Exec);
  \draw[<-] (Some-Learn) -- (Each-Learn);
  \draw[<-] (Some-Exec) -- (Each-Exec);
  \draw[<-] (Each-Learn) -- (Each-Exec);
  \draw[<-] (Some-Exec.-56) [rounded corners] |- (Resp);

  \draw[dashed,<-] (Each-Exec) -- (Resp);

  \end{tikzpicture}
\end{center}
\podc{\Vex{-2}}\disc{\Vex{-2}}\arxiv{\Vex{-1}}
\caption{Hierarchy of assumptions about servers (left) and hierarchy of liveness assertions (right) of 18 algorithms, as stated by the authors of the algorithms, comprising of (i) the 16 algorithms and variants in the last two rows in Table~\ref{tab-link}, and (ii) 2 algorithms from the first row---Derecho and VR-Revisit. Others in Table~\ref{tab-algo} do not state liveness assumptions or assertions. There are 17 entries on the left, because Mencius and Paxos-Complex do not provide server assumptions, and Paxos-SB describes two variants of assumptions. There are 18 entries on the right, but Derecho and Mencius only mentions "liveness" without explanations, IronRSL proves two properties, and Paxos-SB %
  states two properties.
} %
\label{fig-live-hierarchy}
\end{figure*}

\section{Liveness properties and analysis results}
\label{sec-anal}

This section describes liveness properties stated for the algorithms and
variants considered, 
and their proofs and discussions given by their respective
authors.  We then present our analysis of these liveness properties and of
the best possible liveness properties.

\subsection{Liveness properties stated and their proofs given}

We organize liveness properties by the assumptions and assertions stated by
the authors, for all the algorithms and variants shown in
Figure~\ref{fig-live-hierarchy}, 18 total.  Table~\ref{tab-live-summary}
displays these in columns ``Assumptions'' and ``Assertions''.
Column ``Proofs'' states the kinds of proofs of liveness properties given
by the authors of the algorithms and variants.  We categorize the proofs
into the following four kinds:
\begin{enumerate}

\item \textbf{Sketch} --- a proof sketch is given explaining the idea but
  many details are omitted.
\item \textbf{Prose} --- a proof in ordinary prose is given. The proof is
  usually provided in paragraphs and may omit some details.
\item \textbf{Systematic} --- a proof written in the step-wise hierarchical
  structure~\cite{lamport2012write}.  The proof is more systematic than
  ordinary prose proofs.
\item \textbf{Formal ($T$)} --- a proof that is machine-checked, using the
  proof system $T$.

\end{enumerate}
Note that 7 of the 18 algorithms and variants do not have proofs of any
kind given with the stated properties.

\providecommand{\propnameref}[1]{\nameref{#1} (Prop~\ref{#1})}
\providecommand{\cornameref}[1]{\nameref{#1} (Cor~\ref{#1})}
\providecommand{\tableSummaryExistText}{
\begin{tabular}{@{\arxiv{\Hex{-2.7}}\disc{\Hex{-2}}}r@{\,}|@{\,}p{\arxiv{13.5}\disc{17.4}ex}@{} || @{\,}l@{\,}|@{\,}l@{\,} | @{\,}p{15.5ex}@{\,} | @{\,}p{11ex}@{\,} || @{\,}p{32.5ex}@{}}
  & \multirow{2}{*}{~~~~Name} & \multicolumn{2}{c|@{\,}}{Assumptions} 
  & \multirow{2}{*}{~~~Assertions} & \multirow{2}{*}{~~\,Proofs} 
  & \multirow{2}{*}{~~~~~~~~~~~~Analysis}\\
  \cline{3-4}
  & & ~Link & ~~~~~~~Server & &\\
  \hline\hline
    
  3 & EVS & \Fair & \QAlw & \EachExec & Systematic 
    & \propnameref{prop-live-sb}\\
  4 & Paxos-VS & \Fair & \PAlwQ & \Resp & - 
    & \propnameref{prop-live-paxos-vs}\\
  5 & Derecho & \Fair & \PAlwQ & ``progress'' & - & lacking assertions\\
  \hline\hline
  7 & VR-Revisit & \Raw & \PAlwQ & \Resp & Prose 
    & \propnameref{prop-live-raw}\\
  \hline\hline
  8 & Paxos-Synod & \Sure & \PQExtraDur & \EachLearn\m{\!}+ & Prose
    & assuming primary\\
  10 & Paxos-Fast & \Fair & \PQAlw & \SomeLearn & Systematic
     & \cornameref{prop-live-pqalw-too-strong}\\
  \hline\hline
  12 & CT & \Fair & \QAlw & \EachLearn & Systematic 
     & \propnameref{thm-live-act}\\
  13 & ACT & \Fair & \QAlw & \EachLearn & Systematic 
     & \propnameref{thm-live-act}\\
  \hline\hline
  14 & Paxos-Time & \Sure & \PQDur & \EachExec\m{\!}+ & Systematic 
     & assuming primary\\
  \hline
  17 & Chubby-Live & \Sure & \Alw & \Resp & - 
     & \cornameref{prop-live-pqalw-too-strong}, \newline trivial from \Alw\\
  18 & Paxos-SB & \Fair & \QAlw & \SomeExec & - & \propnameref{prop-live-sb}\\
  19 & Mencius & \Fair & - & ``liveness'' & - & lacking assumptions\\
  20 & Zab & \Sure & \PQAlw & \EachExec & Sketch 
     & \cornameref{prop-live-pqalw-too-strong}\\
  22 & EPaxos & \Fair & \AlwQ & \SomeExec & - 
     & \propnameref{prop-live-SomeDec-imp}\\
  24 & Paxos-Complex  & \Fair & - & \Resp & - 
     & lacking assumptions, \newline \propnameref{thm-live-complex}\\
  \hline
  26 & IronRSL & \Sure & \PQAlw & \SomeExec, \Resp & Formal (Dafny) 
     & \cornameref{prop-live-pqalw-too-strong}\\
  28 & LastVoting-PSync & \Sure & \PQDur & \EachLearn\m{\!}+ & Formal (PSync) 
     & assuming primary\\
  29 & Paxos-EPR & \Fair & \PQAlw & \EachVote & Formal (Ivy) 
     & \cornameref{prop-live-pqalw-too-strong}, \newline \propnameref{prop-live-eachvote-weak}\\
  \hline
\end{tabular}
}
\newcommand{\tableSummaryExistCaption}{%
Assumptions, assertions, and proofs, as provided by the respective authors, together with our analysis, about the liveness properties for 18 algorithms, comprising of (i) the 16 algorithms and variants in the last two rows in Table~\ref{tab-link}, and (ii) 2 algorithms from the first row---Derecho and VR-Revisit.
Others in Table~\ref{tab-algo} do not provide assumptions, assertions, or proofs about liveness.\\
``-'' in Assumptions and Proofs indicates that the information is not provided;
for proofs, it means that the property is only claimed to follow from the assumptions, but no proof was given.\\
``+'' in Assertions indicates that the calculation of the bound after which the assertion would be satisfied is also given.
}
\begin{table*}[htp]
  \small
  \centering
  \tableSummaryExistText
  \disc{\Vex{2}}
  \caption{\tableSummaryExistCaption\disc{\Vex{-4}}}
  \label{tab-live-summary}
\end{table*}

\subsection{Analysis of liveness properties}

We present our analysis results showing the range of issues from lacking
assumptions or too weak assumptions for which no liveness assertions can
hold, to too strong assumptions making it uninteresting to achieve the
assertions.
Column ``Analysis'' of Table~\ref{tab-live-summary} summarizes these results.
For example, 
\begin{itemize}
\item EPaxos assumes \AlwQ, that is, eventually, there is always some
  quorum of servers that is non-faulty.  By
  Proposition~\ref{prop-live-SomeDec-imp}, this is too weak for reaching
  consensus, because each such quorum might not stay non-faulty long enough
  to make progress.

\item Zab assume \PQAlw, that is, eventually, there is a server P and a
  quorum Q of servers, such that always, P is non-faulty and is the primary
  and Q is non-faulty.  Chubby-Live even assumes \Alw, that is, eventually,
  all servers are always non-faulty.  By
  Corollary~\ref{prop-live-pqalw-too-strong}, this is too strong, making it
  trivial to reach consensus.
\end{itemize}

\begin{proposition}[\EachVote insufficient]
\label{prop-live-eachvote-weak}
\EachVote is not strong enough as a liveness assertion.
\end{proposition}\disc{\Vex{-4}}
\begin{proof} %
  \EachVote is insufficient to be considered as an important assertion for liveness---a sentiment even echoed in the liveness properties proved by Lamport in both Paxos-Synod and Paxos-Fast. %
As an example, consider Raft, which may satisfy \EachVote with a quorum voting for a value \m{v}, but not reaching consensus on \m{v}.  This happens with view changes, where a different value may be chosen as shown in Figure~\ref{fig-RaftEachVote}.
\end{proof}

\begin{proposition}[None from \Raw]
\label{prop-live-raw}
No algorithm can satisfy any liveness assertion in
Figures~\ref{fig-live-formal} and~\ref{fig-live-formal-multi} under \Raw.
\end{proposition}\disc{\Vex{-4}}
\begin{proof}
  Under \Raw all messages can be lost and therefore none of the liveness
  assertions in Figures~\ref{fig-live-formal}
  and~\ref{fig-live-formal-multi} can be satisfied.
\end{proof}

\begin{proposition}[All from \QAlw single]
\label{thm-live-act}
There is an algorithm that solves single-value consensus and satisfies all
liveness assertions in Figure~\ref{fig-live-formal} under \Fair and \QAlw.
\end{proposition}\disc{\Vex{-4}}
\begin{proof}[Proof sketch]
This is due to ACT. %
Aguilera et al~\cite[Theorem 5]{aguilera2000failure} proves that ACT 
satisfies \EachLearn.
By the hierarchy in Figure~\ref{fig-live-hierarchy}, it also satisfies
\EachVote and \SomeLearn.

ACT does not specify processes executing values, and nor does it specify
clients.  However, it is straightforward to extend ACT to have them.  If a
process can both learn and execute values, then the process executes upon
learning them. Otherwise, upon learning a value, processes send a
\msg{'learned'} message with the value to all the processes that execute
values.  \Fair link assumption states that these messages will be
eventually received.  Therefore, \SomeExec and \EachExec follow from
\SomeLearn and \EachLearn, respectively.

ACT can be further extended by adding clients. Upon executing a value,
processes send a \msg{'resp'} message with the executed value to all the
clients. \Fair and \SomeExec then imply \Resp.
\notes{
straightforward to do: same process then exec, diff proc: send and receive and exec. \SomeExec and \EachExec %
Also, this algorithm does not discuss clients
and therefore \Resp, but it is easy to extend the algorithm to achieve
\Resp.  send to client after execute, fair links.  }
\end{proof}

\begin{proposition}[All from \QAlw]
\label{prop-live-sb}
There is an algorithm that claims to solve consensus and satisfies all
liveness properties shown in Figure~\ref{fig-live-formal-multi} under \Fair
and \QAlw.
\end{proposition}\disc{\Vex{-4}}
\begin{proof}[Proof sketch]
This is due to Paxos-SB. The authors state that it satisfies \SomeExec and \EachExec. By the hierarchy structure, it also satisfies \EachVote, \SomeLearn, and \EachLearn.

However, if clients do not re-send request messages, Paxos-SB does not satisfy \Resp, because only the server that received a client's request can respond to the client~\cite[Figure 10, Lines D3-D4]{kirsch2008paxostr}. Thus, we can construct an adversarial run in which the server that received the client's request executes the request but fails before replying to the client and all other servers are always non-faulty, thereby maintaining \QAlw. If the client does not re-send request message, it will never receive any response. As a remedy, we must have clients re-send their requests to different servers upon not receiving a response within some timeout. Then, \Resp can be satisfied.
\end{proof}

\notes{
Both algorithms that solve consensus under \QAlw and \Fair use rotating primary paradigm to do leader election - 
}

\begin{proposition}[\Resp from \PAlwQ]
\label{prop-live-paxos-vs}
There is an algorithm that claims to solve consensus and satisfy \Resp %
under \Fair and \PAlwQ.
\end{proposition}\disc{\Vex{-4}}
\begin{proof}[Proof sketch]
This is due to Paxos-VS.
\end{proof}

\begin{corollary}[Weak from \PQAlw]
  \label{prop-live-pqalw-too-strong}
  A liveness property that assumes \PQAlw is a weak property, because the
  assumption is stronger than necessary.
\end{corollary}\disc{\Vex{-4}}
\begin{proof}
  This is because \PQAlw is a stronger assumption than \PAlwQ, yet Paxos-VS
  solves consensus under \PAlwQ.
  ~Note that, as a result, \Alw is even more so---it is too strong making
  it trivial to achieve liveness.
\end{proof}

\begin{proposition}[Not \Resp]
\label{thm-live-complex}
Paxos-Complex cannot satisfy \Resp, even with \Alw.
\end{proposition}\disc{\Vex{-4}}
\begin{proof}[Proof sketch]
In Paxos-Complex, servers may concurrently propose values sent by clients on the same slot. However, for each slot, only one value can be chosen. If the value proposed by a server is not chosen on the slot proposed by the server, the server re-proposes the value on a different slot. This can repeat forever. Thus, Paxos-Complex does not satisfy \Resp.
\end{proof}

\notes{
\begin{proposition}
No algorithm can satisfy \EachVote under \Fair.
\end{proposition}\disc{\Vex{-4}}
\begin{proof}[Proof sketch]
This result is due to~\cite[Theorem 1]{fischer1985impossibility}.
\end{proof}

}

The following impossibility result
has also been discussed by Kirsch and Amir~\cite[Strong
L1]{kirsch2008paxostr}.  It is proved by Keidar and
Shraer~\cite{keidar2006timeliness,keidar2006timelinessTR} in a general
round-based algorithm framework called GIRAF.
\begin{proposition}[None from \AlwQ]
\label{prop-live-SomeDec-imp}
No quorum-based consensus algorithm that executes in rounds can satisfy
\SomeLearn under \Fair and \AlwQ.
\end{proposition}\disc{\Vex{-4}}
\begin{proof}[Proof sketch]
If the algorithm does not use a primary, take the server that has proposed the highest round seen by some quorum as the primary. An adversarial run can be constructed where the messages sent by the primary are delayed long enough for a new primary to be elected before the messages are received. Because every sent message is received, \Fair is satisfied, and because at most one server is faulty at any point, \AlwQ is satisfied.
\notes{
\begin{itemize}
    \item Primary-ful. We assume that the algorithm uses an oracle that chooses a primary server and informs all non-faulty servers about it. Recall that only the primary server can propose values.
    \item Primary-less. During executions of these algorithms, the adversary fails the server which has proposed the highest ballot seen by any quorum.

    The invariant maintained by the adversary, $Adv\_Inv$ is,
    \begin{code}
        q.seen_ballots = \{b : \EACH p in q \HAS p.received m and m.ballot = b\}
        
        \p{Adv\_Inv} =
        quorums seen_ballots = \m{\emptyset} or
        \SOME p in servers \HAS \SOME p.sent m \HAS \SOME q in quorums \HAS
            m.ballot in q.seen_ballots and
            \EACH q2 in quorums \HAS m.ballot \m{\geq} max q2.seen_ballots and
            not p.nf
    \end{code}
    
    We prove $[]Adv\_Inv => []~\EachVote$:
    \begin{enumerate}
        \item First, we consider the case when no quorum has seen a ballot, that is, there is no quorum such that every server in the quorum has seen the same ballot. Because \EachVote needs a quorum of servers to vote on the same value 
        
        \EachVote cannot hold for any value because \EachVote needs 
        \begin{code}
        CASE quorums seen_ballots = \m{\emptyset}
        by
        \end{code}
    \end{enumerate}
\end{itemize}
}
\end{proof}
\notes{
\begin{itemize}
\item EPaxos does not satisfy.
\item .... are invalid.
\item Paxos-SB describes this.
\end{itemize}
}

\notes{
This is subsumed by theorem~\ref{thm-live-SomeDec-AlwQ-Fair}
\begin{proposition}
With any link assumption, including bounded delay, no algorithm without failure detection can satisfy \EachVote (and any of the 5 liveness properties). [Toueg]
\end{proposition}
\begin{itemize}
\item 
Paxos-Synod is not live using PQ-Extra-Dur.  because delayed msg can trigger new leader election, so no decision is made.  
no processes joining leaving does not guarantee no leader election.
also, T1 is not used for anything.

\item line 3 of PQ-Extra-Dur is strange: it's not an assumption of the environment but the algorithm.
\end{itemize}
}

\notes{
\begin{proposition}
No algorithm can satisfy \EachVote without \Sure.
\end{proposition}
\begin{proof}[Proof sketch]
This is implied by the FLP impossibility result~\cite{fischer1985impossibility}.
\end{proof}
}

\notes{
PQ-Extra-Dur and PR-Dur are not sufficient for paxos-synod.  same reason as for 2.

Even Evt-Alw and Q-Evt-Alw are not sufficient for synod.  same reason as for 4

PQ-Dur is not strong enough. need to say that a duration of at least certain length.
summarize different duration.
}

\notes{
\begin{proposition}
To satisfy \Alw, \PQAlw, and \QAlw, a single process making decisions is sufficient.
\end{proposition}
}

\notes{
async round based

paxos-sb does not assume fair links.
sc - paxos-sb assumes no communication failures, so that is our fair links
}

\subsection{What is the right or best liveness property, or properties?}
\label{sec-best}

The link assumptions in Figure~\ref{fig-link-formal} form a 3-element
linear order,
but as Figure~\ref{fig-live-hierarchy} shows, the hierarchy of server
assumptions includes two diamonds, and the hierarchy of liveness assertions
includes a diamond in the middle and a branching at the bottom.
That is, different server assumptions and liveness assertions are not
merely a degree of strength but include a combination of different choices.

The main question remaining is, what is a right or best liveness property
that can be satisfied?  This property must give the best liveness assertion
and make the weakest assumptions among all choices.

\mypar{Responding to clients from reaching consensus}
It is clear, from the perspective of truly solving the consensus problem,
that \Resp is the desired and best liveness assertion in
Figure~\ref{fig-live-hierarchy}-right.  Anything weaker does not guarantee
the responsiveness of the servers to the requests of the clients.

Note that the partial order of liveness assertions could be made into a
lattice, by adding a bottom element, Each-Resp, that is stronger than both
\EachExec and \Resp, to close the branching at the bottom.  Each-Resp would
capture that clients have received not only the desired \msg{'resp'}
messages, but received them from each server in some quorum.  However, this
has no benefit but only incurs extra cost.
In fact, it is clear from Figure~\ref{fig-live-hierarchy}-right that
neither \EachExec nor \EachLearn is needed to ensure \Resp.

As discussed in proofs above, \EachVote is insufficient for liveness, but
once \SomeLearn is satisfied, \SomeExec and \Resp are easy to satisfy with
\Fair.  Indeed, it is known that consensus is reached when \SomeLearn is
satisfied.

\mypar{Weakest assumptions for consensus}
So what are the weakest link and server assumptions that can ensure
\SomeLearn?
As discussed in proofs above, \Raw and \AlwQ are both too weak, but ACT and
Paxos-SB achieve consensus and more with \Fair and \QAlw with systematic
proofs, and Paxos-VS and Derecho claim to achieve with \Fair and \PAlwQ
without proofs.

Additionally, Paxos-Time achieves consensus and more with \Sure and \PQDur,
noted as ``assuming primary'' in Table~\ref{tab-live-summary}.  Is it
possible to achieve consensus without assuming primary, that is, with a new
element, Q-Dur, in Figure~\ref{fig-live-hierarchy} that is weaker than both
\PQDur and \QAlw?  We tend to think so, likely by strengthening analysis
for ACT and Paxos-SB, but there are no existing proofs or even claims for
this.
Finally, with use of duration in Q-Dur, and with \Sure, liveness assertions
should also be strengthened to give upper bounds after which the assertions
would be satisfied, similar to what Paxos-Time achieves.  The ultimate goal
is, of course, the most efficient algorithm with the most precise bound,
and possibly with probabilistic guarantees.
Much is open for future studies.

\notes{
For Paxos-Basic. Since, Lamport never really talks about leader election, we have to assume our own assumptions for leader election. I assume a process considers itself leader after it receives a majority of replies. Following should hold if all processes start concurrently. It may hold otherwise as well\\
\begin{verbatim}
leader = the proposer who believes itself to be the leader
followers = proposers who know they are not the leader because
some other proposer has a higher ballot than them and they know him

FD = Failure Detector, drift ignored for now
FD timeout assuming heartbeat is led by followers (pong detector):
    2dmax

FD timeout assuming heartbeat is led by leader (heartbeat):
    If we assume only a single process considers itself to be the leader:
        If the leader keeps checking if it is still the leader every I time units,
        then timeout should be I time units. I >= 2dmax
    else I think it should be I = 2dmax.
\end{verbatim}
}

\section{Related work and conclusion}
\label{sec-live-related}

\notes{

\mypar{Failure Detectors and Consensus}
The problem of solving consensus with failure detectors has been extensively studied in two major categories - one where servers do not recover after crash~\cite{chandra1996unreliable,chandra1996weakest} and the other where servers may recover after crash~\cite{dolev1996failure,oliveira1997consensus,hurfin1998consensus,aguilera2000failure}.

other work on liveness verification:

Mechanically Verifying the Fundamental Liveness Property of the Chord
Protocol M Filali - Formal Methods 2019
} %

There is a large literature on distributed consensus algorithms and
variants. This work
systematically studies, formalizes, and compares liveness assumptions and
assertions in over 30 prominent consensus algorithms and variants under the
asynchrony caused by process and link failures.
Many of these algorithms do not discuss liveness properties, as described
in Section~\ref{sec-live-hierarchy} and summarized in
Table~\ref{tab-live-summary}. Among those that do, except
Amir~\cite{amir1995replication}, none of them specify liveness properties
formally in their papers, as shown in Appendices~\ref{appendix-live-vs},
\ref{appendix-live-vr}, \ref{appendix-live-paxos}, \ref{appendix-live-fd},
and~\ref{appendix-live-others}\footnote{Ironfleet and EPR contain formal
  liveness properties in web links
  but do not present them in the papers.}.  In~\cite{amir1995replication},
Amir formalizes a liveness property similar to \EachExec and provides a
systematic proof of it for EVS, a group communication based consensus
algorithm. Other such algorithms, including
Congruity~\cite{amir2001congruitytr,amir2002congruity} and
COReL~\cite{keidar1994msthesis,keidar1996efficient}, also study a similar
liveness property.

Kirsch and Amir~\cite{kirsch2008paxos,kirsch2008paxostr} discuss two
different liveness properties, called Weak L1 and Strong L1, and state that
Paxos-SB satisfies Weak L1. They further state that the assumptions of the
other liveness property, Strong L1, are too weak and could not guarantee
liveness. Server assumptions for these properties are shown in
Figure~\ref{fig-live-hierarchy}.

Keidar and Shraer~\cite{keidar2006timeliness,keidar2006timelinessTR}
develop GIRAF, a generalized framework that models consensus algorithms as
round-based. The framework is a generalization of Gafni’s Round-by-Round
Failure Detector (RRFD)~\cite{gafni1998rrfd}. The framework is formally
specified in IOA~\cite[Algorithm 1]{keidar2006timeliness}.
Computation in this model proceeds in rounds and, in each round, processes
may query an oracle. They give algorithms that satisfy \EachLearn under
different liveness assumptions in their model and also give an
impossibility result for another liveness assumption.

Van Renesse et al.~\cite{van2015vive} explore the similarities and
differences between VR, Paxos-Synod, and Zab, and discuss liveness of these
algorithms with particular emphasis on recovery after process failures both
with and without stable storage.
Van Renesse and Altinbuken~\cite{van2015paxos} describe \Resp as the
``ideal liveness condition'' for Paxos and, like~\cite{kirsch2008paxostr},
detail certain practical considerations that would affect liveness
including bounded clock drift, timeouts, and stable storage.

The language we introduced is powerful, easy to use, and highly effective.
It is more general than
temporal logic~\cite{pnueli1977temporal,huth2004logic} by supporting also
time intervals, similar to many temporal logic extensions,
e.g.,~\cite{chaochen1991calculus,chaochen2004duration}.  It captures
commonly-used predicate logic over time intervals in simple forms that are
easy to use for understanding and reasoning.  It allowed us to easily
express the wide range of liveness properties clearly and precisely as
needed.  The resulting specifications then allowed us to directly and
precisely relate all different combinations of assumptions and assertions
and to find weaknesses of many of them, as shown in Section~\ref{sec-anal}.

In conclusion, 
this paper is a first step in precise specification of the wide variety of
liveness assumptions and assertions used for distributed consensus
algorithms and variants under asynchrony.  Such precise specification is
essential for comparing the algorithms and variants, including new
algorithms as they are developed, and helps better understand all of them.
Much future work is needed to formally verify liveness properties of
consensus algorithms, especially the ones that give stronger liveness
guarantees with weaker assumptions.
Much remains open for precise time complexity analysis of these algorithms,
as discussed in Section~\ref{sec-best}, and of Byzantine consensus
algorithms, including precise probabilistic modeling and reasoning for
understanding the performance of these algorithms analytically.

\podc{
\begin{acks}
\fund
\end{acks}
}

\bibliography{mybib}
\arxiv{
\bibliographystyle{plainurl}
}
\disc{
\bibliographystyle{plainurl}
}
\podc{
\bibliographystyle{ACM-Reference-Format}
}

\appendix
\section{Liveness model checking}
\label{app-mc}

We developed TLA+ specifications of the liveness properties we formalized,
and used TLC, the model checker for TLA+, to model check execution steps to
illustrate liveness patterns for Paxos.
We use the TLA+ specifications of Paxos from~\cite{chand2018simpler}\footnote{Available online at \forlics{\url{https://github.com/lics2020whatslive/LICS2020}}\arxiv{\url{https://github.com/DistAlgo/proofs}
}};
note that 
nondeterminism in the models already allows all failures considered to
happen.
We add a $tick$ variable in the models.  This variable is initialized to 0
and is incremented by 1 with each action execution.  Therefore, its value
corresponds to the number of actions executed by the system.

We show trends and patterns about the liveness properties using two values
about the number of execution steps.
\begin{itemize}

\item Each run of a model takes a \defn{stable duration start}---the
  minimum number of actions after which liveness assumptions hold.
  For Paxos, the liveness assumptions used in the literature, as shown in
  Table~\ref{tab-live-summary}, are \PQExtraDur, \PQDur, and \PQAlw, for
  Paxos-Synod, Paxos-Time, and Paxos-Fast, respectively.
  We make these assumptions hold by (1) ensuring that a leader election
  starts, if not already started, when the value of $tick$ equals the
  stable duration start, and (2) checking that there are no new leader
  election from this point on.

\item In each model, we use TLC to find the value of \defn{stable duration
    length}---the maximum number of actions after which consensus is
  reached.
  For Paxos, the liveness assertions used in the literature, as shown in
  Table~\ref{tab-live-summary}, are \EachLearn, \EachExec, and \SomeLearn,
  for Paxos-Synod, Paxos-Time, and Paxos-Fast, respectively.
  We use the value of $tick$ after which \EachVote holds; \EachVote ensures
  that consensus is reached in Paxos~\cite{lamport2001paxos}, and Paxos
  does not have the infinite behavior of Raft shown in
  Figure~\ref{fig-RaftEachVote}.
  
\end{itemize}
We use TLC to find the value of stable duration length given the value of
stable duration start, and run TLC on models with different number of
servers.
Paxos uses two kinds of servers: proposers and acceptors.  A proposer may
become a leader and propose values, and an acceptor may vote for leaders
and values.  Any majority set of acceptors is a quorum.  To tolerate
failures of 1 proposer and 1 acceptor, there must be at least 2 proposers
and 3 acceptors.

We present experiments and results for single-value Paxos, which we call
Paxos-Basic, because liveness properties using $tick$ for multi-value Paxos
are too expensive to check using TLC.
We run 6 models of Paxos-Basic; for 2, 3, or 4 proposers, in combination
with 3 or 4 acceptors.  Each run finished in a few hours, but larger
numbers are also too expensive to check using TLC.  We name the model with
$i$ proposers and $j$ acceptors $i\text{P}j\text{A}$.

\pgfplotscreateplotcyclelist{mycyclelist}{
red,mark=diamond, dashed,mark options=solid, mark size=4pt\\%
blue,mark=x, dashed,mark options=solid, mark size=3pt\\%
green!50!black,mark=square, dashed, mark options=solid, mark size=4pt\\
red,mark=triangle, dashed,mark options=solid, mark size=4pt\\
blue,mark=+, dashed,mark options=solid, mark size=3pt\\
green!50!black,mark=o, dashed,mark options=solid, mark size=4.5pt\\
}

\newcommand{\plotsscalefactor}{0.85}
\newenvironment{customlegend}[1][]{%
    \begingroup
    \csname pgfplots@init@cleared@structures\endcsname
    \pgfplotsset{#1}%
}{%
    \csname pgfplots@createlegend\endcsname
    \endgroup
}%

\def\addlegendimage{\csname pgfplots@addlegendimage\endcsname}
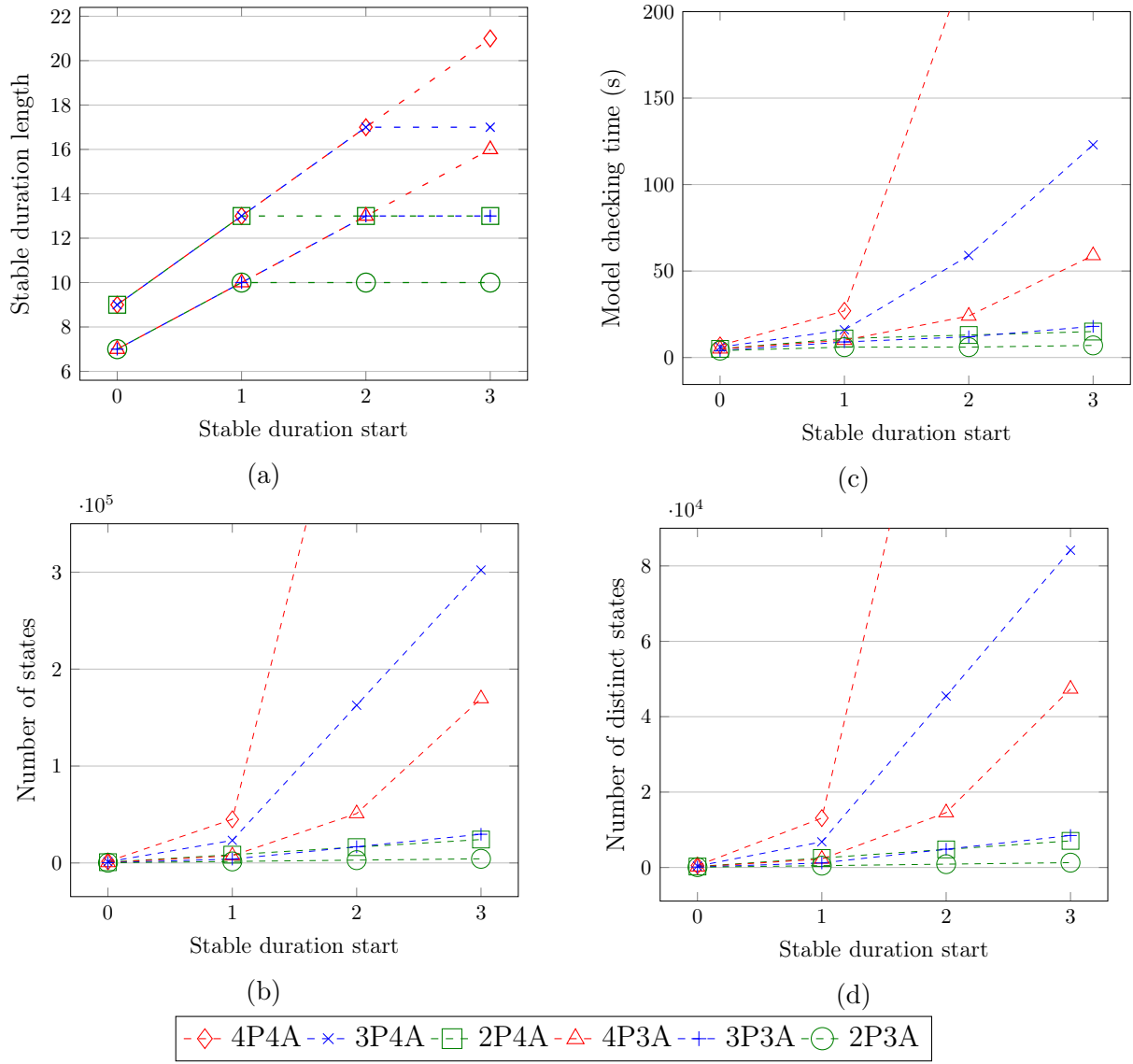
\begin{figure*}[htp]\Vex{-2}
    \begin{subfigure}{0.45\textwidth}
        \centering
        \begin{tikzpicture}[every mark/.append style={mark size=2pt},scale=\plotsscalefactor]
        \begin{axis}[
            cycle list name=mycyclelist,
            xlabel={Stable duration start},
            ylabel={Stable duration length},
            label style={font=\small},
            ylabel near ticks,
            tick label style={font=\footnotesize},
            legend pos=north west,
            legend cell align={left},
            legend style={font=\footnotesize},
            ymajorgrids=true,
            xtick={0, 1, 2, 3},
            width=50ex,
        ]
        
        \addplot+[line width=0.2mm, dash pattern=on 3pt off 6pt, dash phase=6pt] table[x=Start,y=4B4A,col sep=comma] {ModelCheckTicks/PaxosUsDur.csv};
        \addplot+[line width=0.2mm, dash pattern=on 3pt off 6pt, dash phase=3pt] table[x=Start,y=3B4A,col sep=comma] {ModelCheckTicks/PaxosUsDur.csv};
        \addplot+[line width=0.2mm, dash pattern=on 3pt off 6pt ] table[x=Start,y=2B4A,col sep=comma] {ModelCheckTicks/PaxosUsDur.csv};
        \addplot+[line width=0.2mm, dash pattern=on 3pt off 6pt, dash phase=6pt] table[x=Start,y=4B3A,col sep=comma] {ModelCheckTicks/PaxosUsDur.csv};
        \addplot+[line width=0.2mm, dash pattern=on 3pt off 6pt, dash phase=3pt] table[x=Start,y=3B3A,col sep=comma] {ModelCheckTicks/PaxosUsDur.csv};
        \addplot+[line width=0.2mm, dash pattern=on 3pt off 6pt] table[x=Start,y=2B3A,col sep=comma] {ModelCheckTicks/PaxosUsDur.csv};
        
        \end{axis}
        \end{tikzpicture}
        \caption{\label{fig-live-stable-period}%
        }

        \begin{tikzpicture}[every mark/.append style={mark size=2pt},scale=\plotsscalefactor]
        \begin{axis}[
            cycle list name=mycyclelist,
            xlabel={Stable duration start},
            ylabel={Number of states},
            label style={font=\small},
            ylabel near ticks,
            tick label style={font=\footnotesize},
            legend pos=north west,
            legend cell align={left},
            legend style={font=\footnotesize},
            ymajorgrids=true,
            xtick={0, 1, 2, 3},
            width=50ex,
            ymax=350000,
        ]
        \addplot table[x=Start,y=4B4A,col sep=comma] {ModelCheckTicks/PaxosUsStates.csv};
        \addplot table[x=Start,y=3B4A,col sep=comma] {ModelCheckTicks/PaxosUsStates.csv};
        \addplot table[x=Start,y=2B4A,col sep=comma] {ModelCheckTicks/PaxosUsStates.csv};
        \addplot table[x=Start,y=4B3A,col sep=comma] {ModelCheckTicks/PaxosUsStates.csv};
        \addplot table[x=Start,y=3B3A,col sep=comma] {ModelCheckTicks/PaxosUsStates.csv};
        \addplot table[x=Start,y=2B3A,col sep=comma] {ModelCheckTicks/PaxosUsStates.csv};
        
        
        \end{axis}
        \end{tikzpicture}
        \caption{\label{fig-live-states}
        }
        
    \end{subfigure}%
    \hspace{0.05\linewidth}
    \begin{subfigure}{0.45\textwidth}
        \centering

        \begin{tikzpicture}[every mark/.append style={mark size=2pt},scale=\plotsscalefactor]
        \begin{axis}[
            cycle list name=mycyclelist,
            xlabel={Stable duration start},
            ylabel={Model checking time (s)},
            label style={font=\small},
            ylabel near ticks,
            tick label style={font=\footnotesize},
            legend pos=north west,
            legend cell align={left},
            legend style={font=\footnotesize},
            ymajorgrids=true,
            xtick={0, 1, 2, 3},
            width=50ex,
            ymax=200,
        ]
        \addplot table[x=Start,y=4B4A,col sep=comma] {ModelCheckTicks/PaxosUsTime.csv};
        \addplot table[x=Start,y=3B4A,col sep=comma] {ModelCheckTicks/PaxosUsTime.csv};
        \addplot table[x=Start,y=2B4A,col sep=comma] {ModelCheckTicks/PaxosUsTime.csv};
        \addplot table[x=Start,y=4B3A,col sep=comma] {ModelCheckTicks/PaxosUsTime.csv};
        \addplot table[x=Start,y=3B3A,col sep=comma] {ModelCheckTicks/PaxosUsTime.csv};
        \addplot table[x=Start,y=2B3A,col sep=comma] {ModelCheckTicks/PaxosUsTime.csv};
        
        \end{axis}
        \end{tikzpicture}
        \caption{\label{fig-live-time}
        }

        \begin{tikzpicture}[every mark/.append style={mark size=2pt},scale=\plotsscalefactor]
        \begin{axis}[
            cycle list name=mycyclelist,
            xlabel={Stable duration start},
            ylabel={Number of distinct states},
            label style={font=\small},
            ylabel near ticks,
            tick label style={font=\footnotesize},
            legend pos=north west,
            legend cell align={left},
            legend style={font=\footnotesize},
            ymajorgrids=true,
            xtick={0, 1, 2, 3},
            width=50ex,
            ymax=90000,
        ]
        \addplot table[x=Start,y=4B4A,col sep=comma] {ModelCheckTicks/PaxosUsDStates.csv};
        \addplot table[x=Start,y=3B4A,col sep=comma] {ModelCheckTicks/PaxosUsDStates.csv};
        \addplot table[x=Start,y=2B4A,col sep=comma] {ModelCheckTicks/PaxosUsDStates.csv};
        \addplot table[x=Start,y=4B3A,col sep=comma] {ModelCheckTicks/PaxosUsDStates.csv};
        \addplot table[x=Start,y=3B3A,col sep=comma] {ModelCheckTicks/PaxosUsDStates.csv};
        \addplot table[x=Start,y=2B3A,col sep=comma] {ModelCheckTicks/PaxosUsDStates.csv};
        
        \end{axis}
        \end{tikzpicture}
        \caption{\label{fig-live-distinct-states}
        }
    \end{subfigure}

    \begin{subfigure}{0.95\textwidth}
    \centering
    
    \begin{tikzpicture}
        \begin{customlegend}[legend entries={4P4A,3P4A,2P4A,4P3A,3P3A,2P3A}, legend columns=-1]
        \addlegendimage{red,mark=diamond, dashed,mark options=solid, mark size=4pt}
        \addlegendimage{blue,mark=x, dashed,mark options=solid, mark size=3pt}
        \addlegendimage{green!50!black,mark=square, dashed, mark options=solid, mark size=4pt}
        \addlegendimage{red,mark=triangle, dashed,mark options=solid, mark size=4pt}
        \addlegendimage{blue,mark=+, dashed,mark options=solid, mark size=3pt}
        \addlegendimage{green!50!black,mark=o, dashed,mark options=solid, mark size=4.5pt}
        \end{customlegend}
    \end{tikzpicture}
    \end{subfigure}\Vex{2}

    \caption{Model checking execution steps for Paxos-Basic.
      (\subref{fig-live-stable-period}) stable duration start and length as
      measured by TLC 
      (\subref{fig-live-states}) number of states generated by TLC.
      (\subref{fig-live-time}) model checking time (in seconds) taken by TLC.
      (\subref{fig-live-distinct-states}) number of distinct states
      generated by TLC.}
    \label{fig-live-mcplots}
\end{figure*}

\mypar{Results}
Figure~\ref{fig-live-mcplots} shows our model checking results using line
graphs, where each line represents a different combination of numbers of
proposers and acceptors:
\begin{itemize}
\item green for 2 proposers --- lower line (circle marks) for 3 acceptors,
  and higher line (square marks) for 4.
\item blue for 3 proposers --- lower line (plus marks) for 3 acceptors, and
  higher line (cross marks) for 4.
\item red for 4 proposers --- lower line (triangle marks) for 3 acceptors,
  and higher line (diamond marks) for 4.
\end{itemize}
In the graphs for numbers of states and distinct states and model checking
time, the line for 4P4A is cut in order to make the differences among the
other lines visible.
\begin{enumerate}

\item 

  Figure~\ref{fig-live-stable-period} shows the relationship between stable
  duration start and stable duration length.  We see that, for each model,
  stable duration length grows linearly with stable duration start up to an
  inflection point, and remains unchanged after this point.
  \begin{itemize}

  \item Runs with a later stable duration start allow more proposers in
    Paxos to compete to be the leader, and thus take a longer time to reach
    consensus, i.e., have a larger stable duration length.  This is up to
    an inflection point, because the number of proposers in a run is fixed.

  \item Models with fewer proposers reach the inflection point earlier,
    because fewer proposers have less competition and reach consensus
    faster.

  \item Additionally, models with fewer acceptors reaches consensus faster,
    i.e., have a shorter stable duration length, because with majority
    quorums for reaching consensus, 2 votes are needed for 3 acceptors, and
    3 votes are needed for 4 acceptors.

\end{itemize}
The rate at which stable duration length grows, when it does, equals the
number of acceptors. The precise formula is,
\DeclarePairedDelimiter\ceil{\lceil}{\rceil}
\[
  y = 
  \begin{dcases}
    jx + (j + 2 + \ceil*{\frac{j+1}{2}}),& \text{if } x \leq i\\
    ji + (j + 2 + \ceil*{\frac{j+1}{2}}),& \text{otherwise}
  \end{dcases}
\]
where $y$ is stable duration length and $x$ is stable duration start. The
term $(j + 2 + \ceil*{\frac{j+1}{2}})$ is the stable duration length when
stable duration starts at tick 0.
     
\item 

Figures~\ref{fig-live-states} and~\ref{fig-live-distinct-states} show the
number of states and the number of distinct states, respectively, generated
by TLC for each model and stable duration start value.  
For models with 4 proposers, we see that the number of states (and distinct
states) generated by TLC grows exponentially as stable duration start
increases.
For models with 3 proposers, we see a subexponential but superlinear
growth. 
For models with 2 proposers, the growth is approximately linear for number
of states, and is exactly linear for the number of distinct states.
\notes{ distinct states: 2P3A: y = 400x + 81 2P4A: y = 2268x + 253
    
  4P3A: y = -2167.512 + 1758.645e^(1.113082x)
  4P4A: y = -11603.03 + 6476.75 e^(1.615428x)
}
    
\item 

Figure~\ref{fig-live-time} shows the running time (in seconds) that TLC
took to finish checking each model.
For models with 4 proposers, we again see that the growth of model checking
time is exponential. However, for models with 3 proposers, unlike the
trends for number of states and number of distinct states, we see that the
growth of model checking time is superlinear and larger.  For models with 2
proposers, the growth appears to be linear.

\end{enumerate}
In summary, after the liveness assumptions hold, consensus is reached
within a fixed duration determined by the number of proposers and number of
acceptors.  Furthermore, consensus is reached faster when liveness
assumptions hold earlier.

\section{Virtual Synchrony and variants}
\label{appendix-live-vs}
\subsection[VS-ISIS]{VS-ISIS~\cite{birman1987reliable}}
\label{appendix-live-vs-isis}
This work~\cite{birman1987reliable} describes the fundamental communication and maintenance primitives used in the ISIS system. These primitives help build fault-tolerant process groups. The system ensures that the processes belonging to a fault-tolerant process group will observe consistent orderings of events affecting the group as a whole, including process failures, recoveries, migration, and dynamic changes to group properties like member rankings.

The system supports three modes of communication with varying ordering constraints:
\begin{itemize}
    \item \textbf{GBCAST.} The order in which GBCASTs are delivered relative to the delivery of all other sorts of broadcasts (including other GBCASTs) is the same at all overlapping destinations. Additionally, if a process is observed to fail, a failure GBCAST message notifying its failure must be delivered after any other messages
    
    \item \textbf{ABCAST.} These are labeled broadcasts used for maintaining a replicated queue. The order in which ABCASTs are delivered is the same for every intended recipient per label but is independent of other labels and broadcasts.
    
    \item \textbf{CBCAST.} These are similar to ABCAST with the added property that these broadcasts are not independent of other labels. For example, if a replicated variable $x$ is being maintained, and commands setting $x$ to $0$ and incrementing $x$ are to be broadcast in that order, not only is it important that all processes receive the same order of commands but there must be a way to ensure the order.
\end{itemize}

The authors then detail their implementation that guarantees these primitives and give informal correctness (safety) proof for each implementation. For each primitive, atomicity is proven, i.e., every message is either delivered to all or none of its recipients. Then, the particular ordering property for each broadcast is proven. The proofs are manual pen-and-paper proofs. However, the exact properties being proven are not clearly stated all the times. Liveness is neither proven nor clearly stated. Following is the excerpt from the paper on liveness:

\begin{tcolorbox}[title=On Liveness]
In the interest of brevity, we omit a formal proof that the protocols given above
are free of deadlock and livelock.
\end{tcolorbox}

\subsection[VS-ISIS2]{VS-ISIS2~\cite{birman1987exploiting}}
This work describes $ISIS_{2}$. Built upon $ISIS$, this system also supports the same broadcast primitives: GBCAST, ABCAST, and CBCAST. This paper essentially details some parts of the systems and explains some optimizations. While performance is discussed, safety and liveness are not. Following are the assumptions about the system:

\begin{tcolorbox}[title=Assumptions]
In this work, we assume that a distributed system consists of processes with disjoint address spaces
communicating over a conventional LAN using message passing. Processes are assumed to execute on
computing sites. Individual processes and entire sites
can crash; the former type of crash is assumed detectable by some monitoring mechanism at the site of the
process, while the latter can only be detected by
another site by means of a timeout. It is assumed
that failing processes send no incorrect messages.
Our system tolerates message loss, but not partitioning failures (wherein links that interconnect groups of
sites fail). Partitioning could cause parts of our system to hang until communication is restored.
\end{tcolorbox}

\notes{
So, the difference between this and Extended VS is the addition of operation under system partition.

\begin{tcolorbox}[title=Design goal]
It will appear to any observer -- any process
using the system -- that all processes observed
the \textit{same events in the same order}. This applies
not just to message delivery events, but also to
failures, recoveries, group membership changes,
and other events described below. As we will see
in the next section, this enables one to make a-priori assumptions about the actions other
processes will take, and simplifies algorithmic
design.
\end{tcolorbox}

Above is what the authors want a virtually synchronous execution to hold.

To achieve this, Isis toolkit supports many tools. We start with the most low-level primitives supported by the toolkit, the atomic multicast primitives:

\begin{tcolorbox}[title=Atomic Multicast Primitives]
\begin{description}
\item[CBCAST] A commonly occurring
situation involves a number of concurrently executing
processes that communicate with a shared distributed
resource, whose internal state is sensitive to the order
in which requests arrive at different components of
the resource. For example, concurrent operations on
a shared replicated FIFO queue must be received and
processed at all copies in the same order. This
ordering requirement corresponds to the primitive we
call \textit{ABCAST}, which delivers messages atomically and
in the same order everywhere. If all requests for
queue operations are transmitted using this primitive, the enqueuing operations would look synchronous relative to other such operations on the same
queue.

$$
m \rightarrow m' \Rightarrow \A p: deliver_p (m) \xrightarrow{p} deliver_p(m')
$$
\end{description}
\end{tcolorbox}
}

\subsection[EVS]{EVS~\cite{amir1995replication,moser1994extended}}
Extended Virtual Synchrony (EVS) extends Virtual Synchrony by adding more delivery primitives and includes mechanisms to handle system partition.

\begin{tcolorbox}[title=Notation]
\begin{itemize}
\item $S$ is the servers group.
\item $a_{s, i}$ , is the $i^{th}$ action performed by server $s$.
\item $D_{s, i}$ , is the state of the database at server $s$ after actions $1..i$ have been performed by server $s$.
\item $stable\_system(s, r)$ is a predicate that denotes the existence of a set of servers containing $s$ and $r$, and a time, from which on, that set does not face any communication or server failure. Note that this predicate is only defined to reason about the liveness of certain protocols. It does not imply any limitation on our practical protocol.
\end{itemize}
\end{tcolorbox}
\notes{
Using $stable\_system(s, r)$, we can write \textit{Eventual Replication} from~\cite{kirsch2008paxos} as
\begin{verbatim}
\A s, r \in Proc, v \in Values:
    /\ stable_system(s, r)
    /\ Execute(s, v)
    => <> Execute(r, v)
\end{verbatim}
}
\begin{tcolorbox}[title=Liveness property]
If server $s$ performs an action and there exists a set of servers containing $s$
and $r$, and a time, from which on, that set does not face any communication or
processes failures, then server $r$ eventually performs the action.
\end{tcolorbox}

\begin{tcolorbox}[title=Formalization]
$\Diamond (\exists a_{s, i} \land\, \Box stable\_system(s, r)) \Rightarrow \Diamond \exists a_{r, i}$
\end{tcolorbox}
\notes{
I think there is some mismatch between the English definition of $stable\_system(s, r)$ and the use of this predicate in the TLA formula because the definition itself contains the "eventually always" condition.

Events in EVS:
\begin{tcolorbox}[title=EVS Events]
\begin{itemize}
\item $deliver\_conf_p(c)$: the GC delivers to process $p$ a configuration change message initiating configuration $c$ where $p$ is a member of $c$.
\item $send_p(m, c)$: the GC sends message $m$ generated by $p$ while $p$ is a member of configuration $c$.
\item $deliver_p(m, c)$ the GC delivers message $m$ to $p$ while $p$ is a member of configuration $c$.
\item $crash_p(c)$: process $p$ crashes or the processor at which $p$ resides crashes while $p$ is a member of configuration $c$.
\end{itemize}
\end{tcolorbox}

Events in EVS are ordered using two relations:
\begin{tcolorbox}[title=Event ordering relations]
\begin{itemize}
\item The precedes relation, $->$, defines a global partial order on all events in the system.
\item The $ord$ function, from events to natural numbers defines a logical total order on those events. The $ord$ function is not one-to-one, because some events in different processes are required to occur at the same logical time.
\end{itemize}
\end{tcolorbox}

Next the thesis describes the semantics of these two relations. The descriptions are given in English(italicized in the boxes) and a semi-formal formalization is given(normal font in the boxes). The numbering is kept intact from the thesis~\cite{amir1995replication}:

\begin{tcolorbox}[title=Basic Delivery]
\begin{description}
\item [1.1]  \textit{$\rightarrow$ relation is an irreflexive, anti-symmetric and transitive partial order relation.}\\\\
    For any event $e$, $e \not-> e$.\\
    If there exist events $e$ and $e'$, such that $e -> e'$, it is not the case that $e' -> e$.\\
    If there exist events $e$, $e'$ and $e''$ such that $e -> e'$ and $e' -> e''$, then $e -> e''$.

\item [1.2]  \textit{events within a single process are totally ordered by the $\rightarrow$ relation.}\\\\
    If there exists an event $e$ that is $deliver\_conf_p(c)$ or $send_p(m, c)$ or $deliver_p(m, c)$ or $crash_p(c)$, and an event $e'$ that is $deliver\_conf_p(c')$ or $send_p(m', c')$ or $deliver_p(m', c')$ or $crash_p(c')$ , then $e -> e'$ or $e' -> e$.

\item [1.3]  \textit{the sending of a message precedes its delivery, and the delivery occurs in the configuration in which the message was sent or in an immediately following transitional configuration.}\\\\
    If there exists $deliver_p(m, c)$, then there exists $send_q(m, reg(c))$ such that $send_q(m, reg(c)) -> deliver_p(m, c)$.

\item [1.4] \textit{a given message is not sent more than once and is not delivered in two different configurations to the same process.}\\\\
    If there exists $send_p(m, c)$, then $c = reg(c)$ and there is neither $send_p(m, c')$ where $c \neq c'$, nor $send_q(m, c'')$ where $p \neq q$.\\
    Moreover, if there exists $deliver_p(m, c)$, then there does not exist $deliver_p(m, c')$ where $c \neq c'$.
\end{description}
\end{tcolorbox}

\begin{tcolorbox}[title=Delivery of Configuration Changes]
\begin{description}
\item [2.1] \textit{ if a process crashes or partitions, then the GC detects that and delivers a new configuration change message to other processes belonging to the old configuration.
}\\\\
    If there exists $deliver\_conf_p(c)$ and there does not exist $crash_p(c)$ and there does not exist $deliver\_conf_p(c')$ such that $deliver\_conf_p(c) -> deliver\_conf_p(c')$, and if $q$ is a member of $c$, then there exists $deliver\_conf_q(c)$, and there does not exist $crash_q(c)$, and there does not exists $deliver\_conf_q(c'')$ such that $deliver\_conf_q(c) -> deliver\_conf_q(c'')$.

\item [2.2] \textit{ at any moment a process is a member of a unique configuration whose events are delimited by the configuration change event(s) for that configuration.
}\\\\
    If there exists an event $e$ that is either $send_p(m, c)$ or $deliver_p(m, c)$ or $crash_p(c)$, then there exists $deliver\_conf_p(c)$ such that $deliver\_conf_p(c) -> e$, and there does not exist an event $e'$ such that $e'$ is $crash_p(c)$ or $deliver\_conf_p(c')$ and $deliver\_conf_p(c) -> e' -> e$.

\item [2.3] \textit{ an event that follows delivery of a configuration change to one process must also follow delivery of that configuration change to other processes.
}\\\\
    If there exist $deliver\_conf_p(c)$, $deliver\_conf_q(c)$ and $e$, such that $deliver\_conf_p(c) -> e$, then $deliver\_conf_q(c) -> e$.

\item [2.4] \textit{ an event that precedes delivery of a configuration change to one process must also precede delivery of that configuration change to other processes.
}\\\\
    If there exist $deliver\_conf_p(c)$, $deliver\_conf_q(c)$ and $e$, such that $e -> deliver\_conf_p(c)$, then $e -> deliver\_conf_q(c)$.
\end{description}
\end{tcolorbox}

\begin{tcolorbox}[title=Self Delivery]
\begin{description}
\item [3] \textit{ each message that is generated by a process is delivered to this process, provided that it does not crash. Moreover, the message is delivered in the same configuration it was sent, or in the transitional configuration which follows.
}\\\\
    If there exist $send_p(m, c)$ and $deliver\_conf_p(c')$ where $c' \neq trans_p(c)$, such that $send_p(m, c) -> deliver\_conf_p(c')$, and there does not exist $crash_p(com_p(c))$, then there exists $deliver_p(m, com_p(c))$.

\end{description}
\end{tcolorbox}

\begin{tcolorbox}[title=Failure Atomicity]
\begin{description}
\item [4] \textit{ if any two processes proceed together from one configuration to the next, the GC delivers the same set of messages to both processes in that configuration.
}\\\\
    If there exist $deliver\_conf_p(c)$, $deliver\_conf_p(c''')$, $deliver\_conf_q(c)$, $deliver\_conf_q(c''')$ and $deliver_p(m, c)$, such that $deliver\_conf_p(c) -> deliver\_conf_p(c''')$, and there does not exist $deliver\_conf_p(c')$ such that $deliver\_conf_p(c) -> deliver\_conf_p(c') -> deliver\_conf_p(c''')$ and there does not exist $deliver\_conf_q(c'')$ such that $deliver\_conf_q(c) -> deliver\_conf_q(c'') -> deliver\_conf_q(c''')$ then there exists $deliver_q(m, c)$.

\end{description}
\end{tcolorbox}

\begin{tcolorbox}[title=Causal Delivery]
\begin{description}
\item [5] \textit{ if one message is sent before another in the same configuration and if the GC delivers the second of those messages, then it also delivers the first.
}\\\\
    If there exist $send_p(m, c)$, $send_q(m', c)$ and $deliver_r(m', com_r(c))$ such that $send_p(m, c) -> send_q(m', c)$, and there exists $deliver_r(m, com_r(c))$, such that $deliver_r(m, com_r(c)) -> deliver_r(m', com_r(c))$.
\end{description}
\end{tcolorbox}

\begin{tcolorbox}[title=Agreed Delivery]
\begin{description}
\item [6.1] \textit{ total order should be consistent with the partial order
}\\\\
    If there exist events $e$ and $e'$ such that $e -> e'$, then $ord(e) < ord(e')$.

\item [6.2] \textit{the GC delivers configuration change messages for the same configuration, at the same logical time to each of the processes. Messages are also delivered at the same logical time to each of the processes, regardless of the configuration in which they are delivered.
}\\\\
    If there exist events $e$ and $e'$ that are either $deliver\_conf_p(c)$ and $deliver\_conf_q(c)$ or $deliver_p(m, c)$ and $deliver_q(m, c')$, then $ord(e) = ord(e')$.

\item [6.3] \textit{the GC delivers messages in order to all processes except that, in the transitional configuration there is no obligation to deliver messages generated by processes that are not members of that transitional configuration.
}\\\\
    If there exist $deliver_p(m, com_p(c))$, $deliver_p(m', com_p(c))$, $deliver_q(m', c')$ and $send_r(m, reg(c'))$ such that $ord(deliver_p(m, com_p(c))) < ord(deliver_p(m', com_p(c)))$ and $r$ is a member of $c'$, then there exists $deliver_q(m, com_q(c'))$.
\end{description}
\end{tcolorbox}

\begin{tcolorbox}[title=Safe Delivery]
\begin{description}
\item [7.1] \textit{ if the GC delivers a safe message to a process which is
in a configuration, then the GC delivers the message to each of the processes in that configuration unless the process crash. i.e. even if the network partitions at that point, the message is still delivered. 
}\\\\
    If there exists $deliver_p(m, c)$ for a safe message $m$, then for every process $q$ in $c$ there exists either $deliver_q(m, com_q(c))$ or $crash_q(com_q(c))$.

\item [7.2] \textit{ if the GC delivers a safe message
to any of the processes in a regular configuration, then the GC delivered the configuration change message for that configuration to all the members of that configuration.
}\\\\
    If there exists $deliver_p(m, reg(c))$ for a safe message $m$, then for every process $q$ in $reg(c)$ there exists either $deliver\_conf_q(reg(c))$.
\end{description}
\end{tcolorbox}
}

\subsection[Paxos-VS]{Paxos-VS~\cite{birman2010virtually}}
This work focuses on a dynamic reconfiguration model similar to the view model in VR (Appendix~\ref{appendix-live-vr}) and unifies two widely popular approaches to the problem of consensus---Virtual Synchrony (Appendix~\ref{appendix-live-vs-isis}) and state machine replication, particularly Paxos (Appendix~\ref{appendix-live-paxos}). Many algorithms and variants are proposed exhibiting availability conditions. We focus on one of these algorithms---Virtually Synchronous Paxos. The algorithm bases on the concept of \textit{memberships} similar to views in VR. A primary is elected in each membership. During steady-state operation, client requests received by non-primary servers are forwarded to the primary. Upon receiving a request, the primary assigns it to a slot and broadcasts the request and its slot number to the servers in current membership. If the primary fails, reconfiguration is triggered (instead of electing a new primary for the membership even if a majority of servers are non-faulty). Following is mentioned about liveness:

\begin{tcolorbox}[title = Assumptions]
Our aim is to provide services in asynchronous systems whose set of servers is changed through explicit reconfigurations. We assume an additional set, potentially overlapping, of clients. We do not assume any bounds on
message latencies or message processing times (i.e., the execution environment is asynchronous), and messages
may get lost, re-ordered, or duplicated on the network. However, we assume that a message that is delivered was
previously sent by some live member and that correct members can eventually communicate any message.
\end{tcolorbox}

\begin{tcolorbox}[title= Liveness Assumption]
Throughout the lifetime of a membership $M$, a majority of servers are correct
\end{tcolorbox}

\begin{tcolorbox}[title= Liveness Property]
Every client request is eventually replied to.
\end{tcolorbox}

\subsection[Derecho]{Derecho~\cite{jha2019derecho}}
Derecho is a fast state machine replication algorithm. Like VS-ISIS and Paxos-VS, this algorithm also executes on the principle of process groups and memberships. Unlike Paxos-VS though, where a non-primary server routes client requests to the primary, in Derecho each server appends incoming client requests in a local queue and informs other servers in the current group of this. If the group contains $N$ servers and server $i \in [0, N-1]$ inserts a request in index $k$ of its local queue, then the request is \textit{virtually} assigned the slot number $N*k + i$. For example, if the group has 3 servers, then server $0$ inserts requests in slots $\{0, 3, 6, \ldots\}$, server $1$ inserts requests in slots $\{1, 4, 7, \ldots\}$, and server $2$ inserts requests in slots $\{2, 5, 8, \ldots\}$, that is, servers simultaneously propose requests but always in disjoint slots. Liveness is discussed in~\cite[Appendix C]{jha2019derecho}. Following is the liveness assumption:

\begin{tcolorbox}[title= Liveness Assumption]
Derecho will make progress if (1) no more than a minority of the current view fails and (2) the failure detector only reports a failure if the process in question has actually failed \textit{or has become unreachable}.
\end{tcolorbox}

\section{Viewstamped Replication and variants}
\label{appendix-live-vr}
\subsection[VR]{VR~\cite{oki1988viewstamped}}
This work presents the Viewstamped Replication algorithm based on the primary-backup approach. One server is designated as the primary; it executes procedure calls, and participates in two-phase commit. The remaining servers are backups, which are essentially passive and merely receive state information from the primary. Servers run in \textit{views}. Multiple events can cause a server to trigger view change, for example, not receiving timely heartbeat messages from some other server in the view, receiving heartbeat messages from a new server, recovering after a crash, etc. This paper does not mention liveness.

\subsection[VR-Revisit]{VR-Revisit~\cite{liskov2012viewstamped}}
This work gives an updated description of Viewstamped Replication. The description of the core algorithm is simplified and nuances of the application using it are separated out. It explains various optimizations and reconfiguration in detail. Both safety and liveness of the View change, Recovery and Reconfiguration algorithms are discussed and informally proven in this paper.

\begin{tcolorbox}[title=Liveness]
\begin{itemize}
    \item \textit{View change.} The protocol executes client requests provided at least $f + 1$ non-failed replicas, including the current primary, are able to communicate.
    \item \textit{Recovery.} The protocol is live, assuming no more that $f$ replicas fail simultaneously.
\end{itemize}
\end{tcolorbox}

The exact liveness assumptions for Reconfiguration are not explicitly mentioned but it can be seen that the algorithm requires at least $f + 1$ non-failed replicas of the old configuration to be non-faulty for the new configuration to be installed.

\section{Original Paxos and variants}
\label{appendix-live-paxos}
\subsection[Paxos-Synod]{Paxos-Synod~\cite{lamport1998part}}
The seminal Part-time Parliament paper explains Paxos in full and also gives a liveness condition.
\begin{tcolorbox}[title=Presidential selection requirement]
If no one entered or left the Chamber, then after $T$ minutes exactly one priest in the Chamber would consider himself to be the president.
\end{tcolorbox}

\begin{tcolorbox}[title=Liveness Property]
If the presidential selection requirement were met, then the complete protocol would have the property that if a majority set of priests were in the chamber and no one entered or left the Chamber for $T + 99$ minutes, then at the end of that period every priest in the Chamber would have a decree written in his ledger.
\end{tcolorbox}

\begin{tcolorbox}[title=Formalization]
None
\end{tcolorbox}

\begin{tcolorbox}[title=Proof]
Still assuming that $p$ was the only priest initiating ballots, suppose that he were required to initiate a new ballot iff (i) he had not executed step 3 or step 5 within the previous 22 minutes, or (ii) he learned that another priest had initiated a higher-numbered ballot. If the Chamber doors were locked with $p$ and a majority set of priests inside, then a decree would be passed and recorded in the ledgers of all priests in the Chamber within 99 minutes. (It could take 22 minutes for $p$ to start the next ballot, 22 more minutes to learn that another priest had initiated a larger-numbered ballot, then 55 minutes to complete steps 1-6 for a successful ballot.) Thus, the progress condition would be met if only a single priest, who did
not leave the chamber, were initiating ballots.
\end{tcolorbox}

This work assumes bounded message-delay (4 minutes) and bounded action execution time (7 minutes):
\begin{tcolorbox}
They determined that a messenger who did not leave the Chamber would always deliver a message within 4 minutes, and a priest who remained in the Chamber would always perform an action within 7 minutes of the event that caused the action
\end{tcolorbox}

The discussions above mention a president. However, the paper does not specify how a president can be elected:
\begin{tcolorbox}[title=On presidential selection]
The Paxons chose as president the priest whose name was last in alphabetical
order among the names of all priests in the Chamber, though we don’t know exactly
how this was done.\\
\vdots\\
Given the sophistication of Paxon mathematicians, it is widely believed that
they must have found an optimal algorithm to satisfy the presidential selection
requirement. We can only hope that this algorithm will be discovered in future excavations on Paxos.
\end{tcolorbox}

\notes{
Also assumes bounded accuracy on local clocks

"""
The Paxons realized that any protocol to achieve the progress condition must involve measuring the passage of time. The protocols given above for selecting a president and initiating ballots are easily formulated as precise algorithms that set timers and perform actions when time-outs occur—assuming perfectly accurate timers. A closer analysis reveals that such protocols can be made to work with timers having a known bound on their accuracy. The skilled glass blowers of Paxos had no difficulty constructing suitable hourglass timers.
"""

This cannot be specified in TLA because TLA does not have the notion of duration. So, we can't specify natively that for $T$ minutes, no one entered or left the chamber.

A solution is to add the physical time as a dimension to all variables explicitly:

$Liveness(S, T)$: $S$ is the start of the duration. $T$ is the duration.
rule-body:
\begin{verbatim}
    \A t1, t2 \in S..S+T: Proc[t1] = Proc[t2]
\end{verbatim}

rule-head:
\begin{verbatim}
    \A p1, p2 \in leader[S+T]: p1 = p2
\end{verbatim}

The other thing to add the specs is the following operator invoked at every action:
\begin{verbatim}
    AdvanceClocks ==
        \E t \in Nat \ {0}:
            /\ GClock' = GClock + t
                \* /\ \A t2 \in GClock..GClock+t: Proc'[t2] = Proc[GClock]
            /\ \A p \in Proc[GClock]: LClock'[p] = LClock[p] + f(t, p)
    f(t, p) == CHOOSE t2 \in {t}: TRUE
\end{verbatim}
}

\subsection[Paxos-Basic]{Paxos-Basic~\cite{lamport2001paxos}}
The primary focus of this paper is to give a simpler explanation of Paxos and consistency proof, focusing on single-value consensus. Liveness however is not considered formally in this work. Lamport explicitly states in Section 2 of the paper: ``\textit{We won't try to specify precise liveness requirements}''. Following is the informal discussion from the paper:

\begin{tcolorbox}[title=Liveness Discussion]
To guarantee progress, a distinguished proposer must be selected as the
only one to try issuing proposals. If the distinguished proposer can communicate
successfully with a majority of acceptors, and if it uses a proposal
with number greater than any already used, then it will succeed in issuing a
proposal that is accepted. By abandoning a proposal and trying again if it
learns about some request with a higher proposal number, the distinguished
proposer will eventually choose a high enough proposal number.

If enough of the system (proposer, acceptors, and communication network)
is working properly, liveness can therefore be achieved by electing a
single distinguished proposer.
\end{tcolorbox}

\subsection[Paxos-Fast]{Paxos-Fast~\cite{lamport2006fast}}
This paper introduces Fast Paxos, which uses the concept of fast quorums to achieve consensus in one less round than Paxos-Basic~\cite{lamport2001paxos}. Incidentally, this paper discusses liveness of
Paxos-Basic~\cite{lamport2001paxos}.

\begin{tcolorbox}[title = Liveness Assumptions]
An agent is defined to be nonfaulty iff it eventually performs the actions
that it should, such as responding to messages. Define a set G of agents
to be good iff all the agents in G are nonfaulty and, if any one agent in G
repeatedly sends a message to any other agent in G, then that message is
eventually received—more precisely, the message is eventually delivered or
G is eventually not considered to be good. Being nonfaulty and being good
are temporal properties that depend on future behavior.

For any proposer $p$, learner $l$, coordinator $c$, and set $Q$ of acceptors,
define $LA(p, l, c, Q)$ to be the condition that asserts:

\begin{description}
\item[LA1.] $\{p, l, c\} \cup Q$ is a good set.
\item[LA2.] $p$ has proposed a value.
\item[LA3.] $c$ is the one and only coordinator that believes itself to be the
leader.
\end{description}
\end{tcolorbox}

\begin{tcolorbox}[title = Liveness Property]
For any learner $l$, if there ever exists proposer $p$,
coordinator $c$, and majority set $Q$ such that $LA(p, l, c, Q)$ holds
from that time on, then eventually $l$ learns a value.
\end{tcolorbox}

\subsection[Paxos-Vertical]{Paxos-Vertical~\cite{lamport2009vertical}}
This paper describes Vertical Paxos, which adds reconfiguration to Paxos. Liveness is considered beyond the scope of the paper:

\begin{tcolorbox}
The liveness property satisfied by a Vertical Paxos algorithm is similar to
that of ordinary Paxos algorithms, except with the added complication caused by reconfiguration---a complication that arises in any algorithm employing reconfiguration. A discussion of liveness is
beyond the scope of this paper.
\end{tcolorbox}

\section{Consensus with failure detection}
\label{appendix-live-fd}
\subsection[CT]{CT~\cite{chandra1996unreliable}}
This work introduces the concept of unreliable failure detectors and studies how they can be used to solve consensus in asynchronous systems with crash failures (no recovery). Two algorithms are presented which solve consensus using (1) any Strong failure detector and, (2) any Eventual strong failure detector. These failure detectors satisfy strong completeness and (eventual) weak accuracy. Lemma 6.2.2 of the paper defines the liveness property as:

\begin{tcolorbox}[title = Liveness Property]
Every correct process eventually decides some value.
\end{tcolorbox}

A sketch of the proof ($\frac{1}{3}$ page) is given for the algorithm that uses Strong failure detector and a systematic proof (1 page) is given for the algorithm that uses Eventual strong failure detector. Likewise, the liveness assumptions for the two algorithms are different. Algorithm that uses Strong failure detector tolerates at most $n-1$ process crashes where $n$ is the total number of processes. This is not a quorum-based algorithm and therefore not considered in this work. The second algorithm that uses Eventual strong failure detector is quorum-based and its liveness assumption is:

\begin{tcolorbox}[title = Liveness Assumptions]
\begin{itemize}
    \item If we assume that the maximum number of faulty processes is less than half then we can solve Consensus using $\Diamond\mathscr{S}$, a class of failure detectors that satisfy only eventual weak accuracy.
    \item Every nonfaulty process proposes some value.
\end{itemize}
\end{tcolorbox}

\subsection[ACT]{ACT~\cite{aguilera2000failure}}
This work studies the problems of failure detection and consensus in asynchronous systems in which processes may crash and recover (unlike the predecessor CT where crashed processes do not recover), and links may lose messages. New types of failure detectors are developed to suit the crash-recovery model as opposed to prior work that requires unstable processes to be eventually suspected forever. Using one of these failure detectors, the authors develop a consensus algorithm without requiring stable storage and two algorithms with stable storage, but using two different failure detectors. Liveness assumptions and property are as follows:

\begin{tcolorbox}[title = Liveness Assumption]
\begin{itemize}
    \item A majority of processes are good.
    \item Every good process proposes a value.
\end{itemize}
\end{tcolorbox}

\begin{tcolorbox}[title = Liveness Property]
Every good process eventually decides some value.
\end{tcolorbox}

A process is called good if it is either always up or eventually always up. Detailed systematic proofs are given for each of the algorithms.

\notes{
\section{Liveness in generalized frameworks}

\subsection[GIRAF]{GIRAF~\cite{keidar2006timeliness,keidar2006timelinessTR}}
This paper~\cite{keidar2006timeliness} and its full technical report~\cite{keidar2006timelinessTR} develop a generalized framework that models consensus algorithms as round-based. The framework, called GIRAF (General Round-based Algorithm
Framework), is a generalization of Gafni’s Round-by-Round Failure Detector (RRFD)~\cite{gafni1998rrfd}. The framework is formally specified in IOA~\cite[Algorithm 1]{keidar2006timeliness}. Computation in this model proceeds in rounds and in each round processes may query their \textit{failure detector oracle}. In the context of their framework, Keidar and Shraer give two consensus algorithms
}

\section{Other consensus algorithms and variants}
\label{appendix-live-others}
\subsection[Paxos-Time]{Paxos-Time~\cite{prisco00revisit}}
This paper performs a time analysis of a specification of Paxos-Basic ($S_{PAX}$) and Paxos in IOA. This work gives upper bounds on the running time and number of messages assuming upper bounds on message delay and action execution time.

\begin{tcolorbox}[title=Liveness]
Let $\alpha$ be a nice execution fragment of $S_{PAX}$ starting in a reachable state and lasting for more than $24l + 10nl + 13d$. Then the leader $i$ executes $Decide(v')_i$ by time $21l + 8nl + 11d$ from the beginning of $\alpha$ and at most $8n$ messages are sent. Moreover by time $24l + 10nl + 13d$ from the beginning of $\alpha$ any alive process $j$ executes $Decide(v')_j$ and at most $2n$ additional messages are sent.
\end{tcolorbox}

Informally, a nice execution fragment is one in which (1) at least a majority of processes are non-faulty, (2) all messages are delivered in a bounded time. The exact definitions, including original markings from the paper, are:

\begin{tcolorbox}[title=Definitions]
\textbf{Definition 2.1.} A step $(s_{k-1},v(t),s_k)$ of a Clock GTA is called \textit{regular} if $s_k.Clock - s_{k-1}.Clock = t$.  [GTA stands for General Timed Automaton. $v(t)$ is the time-passage action that, when executed by the automaton, increments the $Clock$ by an amount of time $t' \geq 0$, independent of the amount $t$. $s_i$ is a state of the automaton.]

\textbf{Definition 2.2.} A timed execution fragment $\alpha$ of a Clock GTA is called \textit{regular} if all the time-passage steps of $\alpha$ are regular.

\textbf{Definition 3.1.} Given a process automaton \texttt{PROCESS}$_i$, we say that an execution fragment $\alpha$ of \texttt{PROCESS}$_i$ is \textit{stable} if process $i$ is either stopped or alive in $\alpha$ and $\alpha$ is regular.

\textbf{Definition 3.2.} Given a channel \texttt{CHANNEL}$_{i,j}$, we say that an execution fragment $\alpha$ of \texttt{CHANNEL}$_{i,j}$ is \textit{stable} if no $Lose_{i,j}$ and $Duplicate_{i,j}$ actions occur in $\alpha$ and $\alpha$ is regular.

\textbf{Definition 3.5.} Given a distributed system $S$, we say that an execution fragment $\alpha$ of $S$ is \textit{stable} if: (i) for all automata \texttt{PROCESS}$_i$ modelling process $i$, $i \in S$ it holds that $\alpha|$\texttt{PROCESS}$_i$ is a stable execution fragment for process $i$; (ii) for all channels \texttt{CHANNEL}$_{i,j}$ with $i,j \in S$ it holds that $\alpha|$\texttt{CHANNEL}$_{i,j}$ is a stable execution fragment for \texttt{CHANNEL}$_{i,j}$.

\textbf{Definition 3.6.} Given a distributed system $S$, we say that an execution fragment $\alpha$ of $S$ is \textit{nice} if $\alpha$ is a stable execution fragment and a majority of the processes are alive in $\alpha$.
\end{tcolorbox}

\subsection[Paxos-PVS]{Paxos-PVS~\cite{kellomaki2004annotated}}
This paper presents a machine-checked safety proof of Paxos-Basic specified in the logic of the PVS theorem prover. Liveness is neither specified or proved.

\subsection[Chubby]{Chubby~\cite{burrows2006chubby}}
This paper describes the authors' experiences with the Chubby lock service, which is intended to provide coarse-grained locking as well as reliable (though low-volume) storage for a loosely-coupled distributed system. It describes the initial design and expected use of the service, and then the actual use and the modifications made to accommodate actual use. Safety and liveness are not discussed.

\subsection[Chubby-Live]{Chubby-Live~\cite{chandra2007paxos}}
This paper details some of authors' experiences with Chubby lock service and its deployment. The authors describe their testing apparatus---a tool that tests the system's safety and liveness:

\begin{tcolorbox}[title=Liveness Property]
\begin{enumerate}
\item \textbf{Safety mode}. In this mode, the test verifies that the system is consistent. However, the system is
not required to make any progress. For example, it is acceptable for an operation to fail to complete
or to report that the system is unavailable.

\item \textbf{Liveness mode}. In this mode, the test verifies that the system is consistent and is making progress. All operations are expected to complete and the system is required to be consistent.
\end{enumerate}

Our tests start in safety mode and inject random failures into the system. After running for a predetermined
period of time, we stop injecting failures and give the system time to fully recover. Then we switch
the test to liveness mode. The purpose for the liveness test is to verify that the system does not deadlock
after a sequence of failures.
\end{tcolorbox}

\begin{tcolorbox}[title=Formalization]
None
\end{tcolorbox}

\begin{tcolorbox}[title=Proof]
None
\end{tcolorbox}

\subsection[Paxos-SB]{Paxos-SB~\cite{kirsch2008paxos,kirsch2008paxostr}}
This work gives low-level pseudocode for Paxos that was 
implemented in C++ by the authors.

\begin{tcolorbox}[title=Stable Set]
A stable set is a set of processes that are eventually alive and
connected to each other, and which can eventually communicate with each other with some (unknown) bounded message delay.
\end{tcolorbox}
\notes{
Definition is not correct. The missing point is "always" or some "long enough duration".

\begin{verbatim}
StableSet(S, \Delta) ==
    /\ S \subseteq Proc
    /\ <> /\ \A p \in S: alive[p]
          /\ \A p1, p2 \in S: msgdelay[{p1, p2}] <= \Delta
    
StableSet_Proposed_Always(S, \Delta) ==
    /\ S \subseteq Proc
    /\ <>[] /\ \A p \in S: alive[p]
            /\ \A p1, p2 \in S: msgdelay[{p1, p2}] <= \Delta
\end{verbatim}
}

\notes{
\begin{tcolorbox}[title=Stable Component]
A stable component is a set of processes that are
eventually alive and connected to each other and for which all the channels to them from all
other processes (that are not in the stable component) are down.
\end{tcolorbox}

\begin{verbatim}
StableComponent(S) ==
    /\ S \subseteq Proc
    /\ <> /\ \A p \in S: alive[p]
          /\ \E \Delta \in Nat: \A p1, p2 \in S: msgdelay[{p1, p2}] <= \Delta
          /\ \A p \in S, q \notin S: msgdelay[{p, q}] = INFINITY
    
StableComponent_Proposed_Always(S) ==
    /\ S \subseteq Proc
    /\ <>[] /\ \A p \in S: alive[p]
            /\ \E \Delta \in Nat: \A p1, p2 \in S: msgdelay[{p1, p2}] <= \Delta
            /\ \A p \in S, q \notin S: msgdelay[{p, q}] = INFINITY
\end{verbatim}

\begin{tcolorbox}[title=Stable Set w/ Partial Crash/Partial Isolation]
Let $S$ be a stable set that can
communicate with bounded delay $\Delta$, and let $max\_stable\_id$ be the maximum server identifier
in $S$. Let $unstable\_crash$ be the set of servers that perpetually crash and recover, and
let $unstable\_partition$ be the set of servers that are eventually alive but whose communication
with at least one member of $S$ is not always bounded at $\Delta$. There is a time after
which (1) $S$ exists and (2) the members of $S$ do not receive any messages from servers in
$unstable\_crash \cup unstable\_partition$ whose identifiers are greater than $max\_stable\_id$.

\end{tcolorbox}

\begin{verbatim}
Unstable_Crash(P) ==
    /\ P \subseteq Proc
    /\ \A p \in P: /\ []<>(~alive[p])
                   /\ []<>(alive[p])

Unstable_Partition(P, S, \Delta) ==
    /\ P \subseteq Proc
    /\ <> \A p \in P: alive[p]
    /\ [] \E s \in S: \A p \in P: msgdelay[{p, s}] > \Delta

StableSet_PCPI(S, \Delta) ==
    <>[] /\ StableSet(S, \Delta)
         /\ \E UC, UP \subseteq Proc:
            /\ Unstable_Crash(UC)
            /\ Unstable_Partition(UP, S, \Delta)
            /\ \A s \in S, p \in UC \cup UP:
                (\A s2 \in S: p.id > s2.id) =>
                (msgdelay[{p, s}] = INFINITY)
\end{verbatim}

\begin{tcolorbox}[title=Stable Set w/ Partial Partition Isolation]
Let $S$ be a stable set that can
communicate with bounded delay $\Delta$, and let $max\_stable\_id$ be the maximum server identifier
in $S$. Let $unstable\_partition$ be the set of servers that are eventually alive but whose communication
with at least one member of $S$ is not always bounded at $\Delta$. There is a time after
which (1) $S$ exists and (2) the members of $S$ do not receive any messages from servers in
$unstable\_partition$ whose identifiers are greater than $max\_stable\_id$.
\end{tcolorbox}
}

\begin{tcolorbox}[title=Paxos-L1 (Progress)]
If there exists a stable set consisting of a majority of servers, then
if a server in the set initiates an update, some member of the set eventually executes the update.
\end{tcolorbox}
\notes{
\begin{verbatim}
Execution(S, v) ==
    \E s \in S: Init_Update(s, v)
    =>
    \E s \in S: <> Execute(s, v)

Paxos_L1_Progress ==
    \E S \subseteq Proc, \Delta \in Nat:
        /\ Cardinality(S) > Cardinality(Proc) \div 2
        /\ StableSet(S, \Delta)
        =>  \E v \in Values: Execution(S, v)
\end{verbatim}
}
The following property states that after the system becomes "nice", if some nice server initiates an update, then some nice server eventually executes it, i.e., consensus was reached and at least one server recognized that.

\begin{tcolorbox}[title=Eventual Replication]
If server $s$ executes an update and there exists a set
of servers containing $s$ and $r$, and a time after which the set does not experience any communication or process failures, then $r$ eventually executes the update.
\end{tcolorbox}
\notes{
\begin{verbatim}
Eventual_Replication ==
    \E s \in Proc, S \subseteq Proc, \Delta \in Nat, v \in Values:
        /\ Execute(s, v)
        /\ s \in S
        /\ Cardinality(S) >= 2
        /\ <>[] /\ \A p1, p2 \in S: ideal_channel(p1, p2)
                /\ \A p \in S: alive[p]
        => \A r \in S: <> Execute(r, v)
\end{verbatim}

I would personally write some/all instead of "r". I believe that they are writing like this because Yair defines $stable\_system(s, r)$ in his thesis to be a set that $s$ and $r$ are a part of and which after some time does not experience any failures.

This property is talking about propagating the results of consensus to a set of alive processes. So, process $s$ learned about the decided value (and executed it), and then it propagates this information to all other processes in some set of alive processes $s$ is a part of.

I think an easier way to write the same property is:

\begin{tcolorbox}[colback=blue!5!white, colframe=blue!75!black, title=Eventual Replication (Saksham)]
If there exists a set of servers, and a time after which the set does not experience any communication or process failures, then after that time, any update executed by a server in the set is eventually executed by all servers in the set.
\end{tcolorbox}
}

The authors claim that the progress property Paxos-L1 can be realized in the following two ways---Strong L1 and Weak L1. Weak L1 assumes that at least a single majority remains alive during the execution, whereas Strong L1 allows majorities to change during the execution such that no majority remains alive during the whole execution. The authors further claim that ensuring progress under Strong L1 is difficult in real world and doubt that any real Paxos-like implementation guarantees that.

\begin{tcolorbox}[title=Strong L1]
If there exists a time after which there is always a set of
running servers $S$, where $|S|$ is at least the majority, then if a server in the set initiates an update, some member of the set eventually executes the update.
\end{tcolorbox}
\notes{
\begin{verbatim}
AliveMajorties ==
    {S \in SUBSET Proc: /\ Cardinality(S) > Cardinality(Proc) \div 2
                        /\ \A p \in S: alive[p]}

Strong_L1_MajoritySet ==
    <>[] AliveMajorities # {}
    =>
    <>[] \E S \in AliveMajorities, v \in Values: Execution(S, v)
\end{verbatim}

Note that in this property the set $S$ appearing in rule-body is not the same as the set in the rule-head. In fact the two sets in the rule-head itself are also not the same. All the sets are one of the $AliveMajorites$ at some time.
}
\begin{tcolorbox}[title=Weak L1]
If there exists a stable component consisting of a majority of servers, and a time after which the set does not experience any communication or process failures, then if a server in the set initiates an update, some member of the set eventually executes the update
\end{tcolorbox}
\notes{
\begin{verbatim}
Weak_L1 ==
    \E S \subseteq Proc, v \in Values:
        /\ StableComponent_Proposed_Always(S)
        /\ Cardinality(S) > Cardinality(Proc) \div 2
        => Execution(S, v)
\end{verbatim}
}
\begin{tcolorbox}[title=Formalization]
None
\end{tcolorbox}

\begin{tcolorbox}[title=Proof]
None
\end{tcolorbox}

Note that Paxos-L1 implies Weak L1. And the authors claim that the leader election protocol in Section 5.3 meets Paxos-L1. The following excerpt from the paper sheds more light as to why they define Weak and Strong L1. It gives their conjecture that any Paxos-like algorithm cannot achieve Strong L1 whereas group communication based protocols can achieve Weak L1:

\begin{tcolorbox}[title=Liveness discussion]
Strong L1 requires that progress be made even in the face of a (rapidly) shifting majority.
We believe that no Paxos-like algorithm will be able to meet this requirement. If the
majority shifts too quickly, then it may never be stable long enough to complete the leader
election protocol. Weak L1, on the other hand, reflects the stability required by many group
communication based protocols, as described above. It requires a stable majority component,
which does not receive messages from servers outside the component. Since the leader
election protocol specified in Section 5.3 meets Paxos-L1, it also meets Weak L1. We note
that group communication based protocols can most likely be made to achieve Paxos-L1 by
passing information from the application level to the group communication level, indicating
when new membership should be permitted (i.e., after some progress has been made).
\end{tcolorbox}

\subsection[Mencius]{Mencius~\cite{mao2008mencius}}
Mencius is a protocol for general state machine replication that
has high performance in a wide-area network. Mencius has high throughput under high
client load and low latency under low client load even under changing wide-area network environment and client
load. Mencius was developed as a derivation from Paxos. Unlike Paxos, it is a multi-leader protocol where each leader proposes values on mutually exclusive slots. For example in a two server system, one server proposes for odd-numbered slots and the other for even-numbered slots. Even though liveness is not explicitly proven, the authors do mention that Mencius assumes eventual delivery for liveness:

\begin{tcolorbox}[title = Liveness Assumption]
Since $\mathcal{P}$ (Mencius is $\mathcal{P}$ combined with certain optimizations) only relies on the eventual delivery of messages for liveness, adding Optimization 2
and Accelerator 1 to protocol $\mathcal{P}$ still implements replicated state machines correctly.
\end{tcolorbox}

\subsection[Zab]{Zab~\cite{junqueira2011zab}}
Zab is an atomic broadcast protocol for primary-backup systems. Zab is shown to guarantee the following liveness property:

\begin{tcolorbox}
Suppose that:
\begin{itemize}
    \item a quorum $Q$ of followers is up;
    \item the followers in $Q$ elect the same process $l$ and $l$ is up;
    \item messages between a follower in $Q$ and $l$ are received in a timely fashion.
\end{itemize}

If $l$ proposes a transaction $\langle v, z \rangle$, then $\langle v, z \rangle$ is eventually committed.
\end{tcolorbox}

They provide an informal proof of liveness that is about 2 paragraphs long.

\subsection[Zab-FLE]{Zab-FLE~\cite{medeiros2012zookeeper}}
This work provides pseudocode of the three phases of Zab---Discovery, Synchronization, and Broadcast phases. It then explains a Fast Leader Election (FLE) algorithm that was adopted in Zab, and replaces the Discovery and Synchronization phases with a single Recovery phase. This optimization was later shown to produce a liveness issue where the system may get caught in an infinite loop~\cite[Section~4.2]{medeiros2012zookeeper}. This paper presents the fix adopted by Zab---re-introduce certain state variables that were removed in the (manual) optimization process. It is then informally argued that the re-introduction of previously maintained state variables solves the issue but no proof is given.

\subsection[EPaxos]{EPaxos~\cite{moraru2013there}}
EPaxos, short for Egalitarian Paxos, can be understood as a variant of Fast Paxos that (1) decides a sequence of commands instead of just one command, and (2) achieves consensus in a leaderless setting where commands may interfere with each other. Liveness is mentioned but not proved:

\begin{tcolorbox}[title = Liveness]
Finally, liveness is guaranteed with high probability as long as a majority of replicas are non-faulty:
clients and replicas use time-outs to resend messages, and a client keeps retrying a command until a replica
succeeds in committing that command.
\end{tcolorbox}

\subsection[Raft]{Raft~\cite{ongaro2014search}}
Raft is another consensus algorithm based on Paxos, inspired by the primary-backup approach of VR. This paper gives a liveness argument for the algorithm stating that the system will continue to make progress as long as the following timing condition is satisfied:

\begin{tcolorbox}[title=Timing condition]
$
broadcastTime \ll electionTimeout \ll MTBF
$\\

In this inequality $broadcastTime$ is the average time it
takes a server to send RPCs in parallel to every server
in the cluster and receive their responses; $electionTimeout$
is the election timeout described in Section 5.2; and
$MTBF$ is the average time between failures for a single
server.
\end{tcolorbox}

However, it is not clear what they exactly mean by progress.

\subsection[Paxos-Complex]{Paxos-Complex~\cite{van2015paxos}}

The paper discusses liveness in Section 3. But does not make clear what liveness means for them. They write the following as an "example of a liveness property", but do not explicitly say that they are invested in achieving this.

\begin{tcolorbox}[title=Liveness Property]
If a client broadcasts a new command to all replicas, it eventually receives at least one response.
\end{tcolorbox}

The paper however discusses liveness in the presence of the following assumptions and how to achieve these assumptions in practice but does not prove liveness:

\begin{tcolorbox}[title=Assumptions]
\begin{enumerate}
\item the clock drift of a process, that is, the rate of its clock is within some [bounded] factor of the
rate of real time
\item the time between when a nonfaulty process initiates sending a message and the message has been received and handled by a nonfaulty destination process [is bounded].
\end{enumerate}
\end{tcolorbox}

The authors also claim that their liveness assumption on messages (Assumption 2 above), can be weakened to fair links assumption. They informally argue that liveness would hold in either case.

However, there is the issue of starvation in the context of the liveness property discussed here. To counter, the authors suggest that \textit{in the absence of new requests}, liveness could be satisfied:

\begin{tcolorbox}[title=Assumptions]
\ldots Should a competing command get decided, the replica
will reassign the request to a new slot and retry. Although this may lead to starvation,
in the absence of new requests, any outstanding request will eventually get decided in
at least one slot.
\end{tcolorbox}

\subsection[Raft-Verdi]{Raft-Verdi~\cite{wilcox2015verdi}} 
This paper presents Verdi---a verification framework for distributed
algorithms---and a mechine-checked proof of linearizability of
Raft as a case study. Liveness is left as future work:

\begin{tcolorbox}[title=Liveness]
  This paper focuses on safety properties for distributed systems; we leave
  proofs of liveness properties for future work.
\end{tcolorbox}

\subsection[IronRSL]{IronRSL~\cite{hawblitzel2015ironfleet}}
IronRSL is part of the IronFleet project of Microsoft. It gives machine-checked \textit{proofs} of one of Microsoft's State machine replication engines based on Paxos. The proofs are in the Dafny system. They prove many liveness properties, most notably:

\begin{tcolorbox}[title = Liveness Property]
  If the network is eventually synchronous for a live quorum of replicas,
  then a client repeatedly submitting a request eventually receives a
  reply.
\end{tcolorbox}

Following are their liveness assumptions (from the github repository linked
from their project page given in the paper):
\begin{tcolorbox}[title = Liveness Assumptions]
\begin{description}
    \item[NoZenoBehavior.] Time goes forward without Zeno behaviors
    \item[ClockAmbiguityLimitedForHosts.] Live replicas have fairly accurate time sync [bound by $asp.max\_clock\_ambiguity$]
    \item[LiveQuorum.] The live quorum is a set of replica indices and is big enough to constitute a quorum
    \item[HostExecutesPeriodically.] Each host in the live quorum executes periodically, with period $asp.host\_period$
    \item[PersistentClientSendsRequestPeriodically.] The persistent client sends its request periodically, with period $asp.persistent\_period$
    \item[NetworkSynchronousForHosts.] The network delivers packets among the client and live quorum within time $asp.latency\_bound$
    \item[NetworkDeliveryRateBoundedForHosts.] The network doesn't deliver packets to any host in the live quorum faster than it can process them
\end{description}
\end{tcolorbox}

\notes{
6/17/21 at end of
https://github.com/microsoft/Ironclad/blob/main/ironfleet/src/Dafny/Distributed/Protocol/RSL/LivenessProof/Assumptions.i.dfy
(difference from notes below: denotation for comments, new ``The request is
valid'', new ``...but it never received a reply'')

///////////////////////////
// TOP-LEVEL ASSUMPTIONS
///////////////////////////

predicate LivenessAssumptions(
  b:Behavior<RslState>,
  asp:AssumptionParameters
  )
{
  && IsValidBehavior(b, asp.c)
  && var eb := RestrictBehaviorToEnvironment(b);
    var live_hosts := SetOfReplicaIndicesToSetOfHosts(asp.live_quorum, asp.c.config.replica_ids);
    && HostQueuesLive(eb)

    // The synchrony period start time is valid
    && 0 <= asp.synchrony_start

    // Time goes forward without Zeno behaviors
    && NoZenoBehavior(eb)

    // Live replicas have fairly accurate time sync
    && 0 <= asp.max_clock_ambiguity
    && ClockAmbiguityLimitedForHosts(eb, 0, asp.max_clock_ambiguity, live_hosts)

    // Overflow protection is never used
    && sat(0, always(OverflowProtectionNotUsedTemporal(b, asp.c.params, asp.max_clock_ambiguity)))

    // The live quorum is a set of replica indices and is big enough to constitute a quorum
    && (forall replica_index :: replica_index in asp.live_quorum ==> 0 <= replica_index < |asp.c.config.replica_ids|)
    && |asp.live_quorum| >= LMinQuorumSize(asp.c.config)

    // Each host in the live quorum executes periodically, with period asp.host_period
    && 0 < asp.host_period
    && (forall replica_index {:trigger HostExecutesPeriodically(b, asp, replica_index)} ::
         replica_index in asp.live_quorum ==> HostExecutesPeriodically(b, asp, replica_index))

    // The request is valid
    && |asp.persistent_request.request| <= MaxAppRequestSize()

    // The persistent client sends its request periodically, with period asp.persistent_period
    && 0 < asp.persistent_period
    && asp.persistent_request.Request?
    && (forall replica_index {:trigger PersistentClientSendsRequestPeriodically(b, asp, replica_index)} ::
         replica_index in asp.live_quorum ==> PersistentClientSendsRequestPeriodically(b, asp, replica_index))

    // and the persistent client never sends a request with a higher sequence number
    && sat(0, always(ClientNeverSentHigherSequenceNumberRequestTemporal(b, asp)))

    // ..but it never receives a reply.
    && sat(0, always(not(ReplyDeliveredTemporal(b, asp.persistent_request))))

    // The network delivers packets among the client and live quorum within time asp.latency_bound
    && 0 < asp.latency_bound
    && NetworkSynchronousForHosts(eb, asp.synchrony_start, asp.latency_bound,
                                  live_hosts + {asp.persistent_request.client}, live_hosts + {asp.persistent_request.client})

    // The network doesn't deliver packets to any host in the live quorum faster than it can process them
    && 0 < asp.burst_size
    && NetworkDeliveryRateBoundedForHosts(eb, asp.synchrony_start, asp.burst_size,
                                          asp.burst_size * asp.host_period * LReplicaNumActions() + 1, live_hosts)
}
}

\notes{
Below is the code from \url{https://github.com/Microsoft/Ironclad/blob/master/ironfleet/src/Dafny/Distributed/Protocol/RSL/LivenessProof/Assumptions.i.dfy}. The comments, colored red, give the English liveness assumptions.

\begin{lstlisting}[breaklines]
///////////////////////////
// TOP-LEVEL ASSUMPTIONS
///////////////////////////

predicate LivenessAssumptions(
    b:Behavior<RslState>,
    asp:AssumptionParameters
    )
{
       IsValidBehavior(b, asp.c)
    && var eb := RestrictBehaviorToEnvironment(b);
       var live_hosts := SetOfReplicaIndicesToSetOfHosts(asp.live_quorum, asp.c.config.replica_ids);
       HostQueuesLive(eb)

    <@\textcolor{red}{// The synchrony period start time is valid}@>
    && 0 <= asp.synchrony_start

    <@\textcolor{red}{// Time goes forward without Zeno behaviors}@>
    && NoZenoBehavior(eb)

    <@\textcolor{red}{// Live replicas have fairly accurate time sync}@>
    && 0 <= asp.max_clock_ambiguity
    && ClockAmbiguityLimitedForHosts(eb, 0, asp.max_clock_ambiguity, live_hosts)

    <@\textcolor{red}{// Overflow protection is never used}@>
    && sat(0, always(OverflowProtectionNotUsedTemporal(b, asp.c.params, asp.max_clock_ambiguity)))

    <@\textcolor{red}{// The live quorum is a set of replica indices and is big enough to constitute a quorum}@>
    && (forall replica_index :: replica_index in asp.live_quorum ==> 0 <= replica_index < |asp.c.config.replica_ids|)
    && |asp.live_quorum| >= LMinQuorumSize(asp.c.config)

    <@\textcolor{red}{// Each host in the live quorum executes periodically, with period asp.host\_period}@>
    && 0 < asp.host_period
    && (forall replica_index {:trigger HostExecutesPeriodically(b, asp, replica_index)} :: replica_index in asp.live_quorum ==> HostExecutesPeriodically(b, asp, replica_index))

    <@\textcolor{red}{// The persistent client sends its request periodically, with period asp.persistent\_period}@>
    && 0 < asp.persistent_period
    && asp.persistent_request.Request?
    && (forall replica_index {:trigger PersistentClientSendsRequestPeriodically(b, asp, replica_index)} :: replica_index in asp.live_quorum ==> PersistentClientSendsRequestPeriodically(b, asp, replica_index))

    <@\textcolor{red}{// and the persistent client never sends a request with a higher sequence number}@>
    && sat(0, always(ClientNeverSentHigherSequenceNumberRequestTemporal(b, asp)))

    <@\textcolor{red}{// ..but it never receives a reply.}@>
    && sat(0, always(not(ReplyDeliveredTemporal(b, asp.persistent_request))))

    <@\textcolor{red}{// The network delivers packets among the client and live quorum within time asp.latency\_bound}@>
    && 0 < asp.latency_bound
    && NetworkSynchronousForHosts(eb, asp.synchrony_start, asp.latency_bound,
                                  live_hosts + {asp.persistent_request.client}, live_hosts + {asp.persistent_request.client})

    <@\textcolor{red}{// The network doesn't deliver packets to any host in the live quorum faster than it can process them}@>
    && 0 < asp.burst_size
    && NetworkDeliveryRateBoundedForHosts(eb, asp.synchrony_start, asp.burst_size,
                                          asp.burst_size * asp.host_period * LReplicaNumActions() + 1, live_hosts)
}
\end{lstlisting}
}
\subsection[Paxos-TLA]{Paxos-TLA~\cite{chand2016formal}}
This work gives a machine-checked proof of safety of a specification of Paxos written in \tlaplus{}. This work does not discuss liveness formally.

\subsection[LastVoting-PSync]{LastVoting-PSync~\cite{druagoi2016psync}}
LastVoting is Paxos in Heard-Of (HO-) model~\cite{charron2009heard}. This protocol proceeds in synchronous rounds where messages from old rounds (except the last one) are discarded. Leader selection in this protocol is intrinsic, based on the current round number $r$, which every alive process has access to. PSync is an automatic verification tool based on HO-model and distributed algorithms written in a certain specification logic called $\mathds{CL}$:
\begin{tcolorbox}[title=Specification logic]
For fault-tolerant distributed algorithms, the specification, the
liveness assumptions, and the inductive invariants, require reasoning about set comprehensions, cardinality constraints, and process
quantification. To express these properties and check their validity
we use a fragment of first-order logic, called $\mathds{CL}$, and its semidecision procedure~\cite{druagoi2014logic}.
\end{tcolorbox}

\notes{Question: What is the relation between CL and EPR? EPR paper mentions CL but doesn't compare it with EPR. I've asked Cezara about this.}
The authors discuss liveness of LastVoting as:
\begin{tcolorbox}[title=Liveness]
In order to ensure termination $LastVoting$ assumes that eventually there exists a
sequence of four rounds, starting with Collect, where the coordinator is in the HO-set of every process, and during the Collect
and Quorum rounds of this sequence the HO-set of each process
contains at least $n/2$ processes.
\end{tcolorbox}

\subsection[Paxos-EPR]{Paxos-EPR~\cite{padon2017paxos,padon2017reducing}}
This work aims at specifying Paxos in a constrained logic, EPR (Effectively Propositional Logic), similar to PSync~\cite{druagoi2016psync,druagoi2014logic}, and then proving its safety~\cite{padon2017paxos} and liveness~\cite{padon2017reducing} properties automatically. The authors prove liveness of Stoppable Paxos~\cite{malkhi2008stoppable} under the following assumptions:

\begin{tcolorbox}[title = Liveness]
Paxos eventually reaches a decision provided that there is a node $l$ and a majority of nodes $Q$ such that
\begin{itemize}
    \item no action of $l$ or of any node in $Q$ can become forever enabled and never executed,
    \item every message sent between $l$ and the nodes in $Q$ is eventually delivered, and 
    \item eventually, no node different from $l$ tries to propose values
\end{itemize}
\end{tcolorbox}

\subsection[Paxos-Decon]{Paxos-Decon~\cite{garcia2018paxos}}
This work aims at specifying and verifying complex Paxos variants by (1) devising composable specifications for implementations of Paxos-Basic, and (2) engineering disciplines to reason about protocol-aware, semantics-preserving optimizations to Paxos-Basic. This work provides safety proof for Paxos-Basic and explains a method to derive Paxos from Paxos-Basic while preserving the safety property. Liveness is not discussed.

\subsection[Paxos-High]{Paxos-High~\cite{liu2019moderately}}
This work aims at giving precise high-level executable specifications of Paxos and its variants, based on lower level pseudo-code in Paxos-Complex~\cite{van2015paxos}. These specifications are written in DistAlgo~\cite{liu2017clarity}. The authors demonstrate how writing high-level specifications led them to find liveness violations, if messages can be lost, in Paxos-Complex---an observation confirmed by the authors of Paxos-Complex. They also demonstrate how they could easily fix the violation using smart retransmission because of the high-level nature of their specifications. Despite this, liveness is neither formally nor informally described or proven.

\notes{
\subsection{VCC}
To the best of my knowledge, Ernie's work is not public yet. But I don't remember any liveness conditions being proven there, only safety.

\subsection{Others}
A divide \& conquer approach to liveness model checking under fairness & anti-fairness assumptions
\url{https://link.springer.com/content/pdf/10.1007\%2Fs11704-017-7036-2.pdf}
\cite{Ogata2018}
}

\notes{
\subsection{Liveness for lamutex~\cite{lamport1978time}}

\begin{tcolorbox}[title=Liveness property]
III. If every process which is granted the resource eventually releases it, then every request is eventually granted.
\end{tcolorbox}

\begin{tcolorbox}[title=Formalization]
None
\end{tcolorbox}

\begin{tcolorbox}[title=Proof]
Rule 2 guarantees that after $P_i$ requests the resource,
condition (ii) of rule 5 will eventually hold. Rules 3 and
4 imply that if each process which is granted the resource
eventually releases it, then condition (i) of rule 5 will
eventually hold, thus proving condition III. 
\end{tcolorbox}

Let us define the following operators:
\begin{itemize}
    \item[rule-body]{This is the antecedent. So in any liveness property, it is the assumption. This defines what it means for a system to be nice}
    \item[rule-head]{This is the consequent. This is what liveness means for the system according to the specifier.}
\end{itemize}

"Ideal" way to specify this given the following operators. Since the processes in the system can make multiple requests without pending requests getting fulfilled, we need a way to say which request was granted and released.
\begin{itemize}
\item $request(p, t)$: Process $p$ requests resource at logical time $t$
\item $granted(p, t)$: Process $p$ granted resource for request made at logical time $t$
\item $release(p, t)$: Process $p$ releases resource for the grant for logical time $t$
\end{itemize}

rule-body: $\A p \in Proc: granted(p, t) ~> release(p, t)$.

$~>$ is leads to operator in \tlaplus{}. $P ~> Q$ implies that if $P$ ever becomes true, at some point afterwards $Q$ must be true. Its short for $[](P => <>Q)$~\cite{lamport1997operators}

rule-head: $\A p \in Proc: request(p, t) ~> granted(p, t)$.

This all can be specified in and for the fair version precisely. But for the exact version, the processes don't release with the timestamp of the request they are release-ing, rather with some timestamp greater than the request's. So if p1 requests at times 1 and 2, and issues a release(5), we can't know which request it was for. This can be fixed by adding a ghost variable, $rel$, that maintains a list of released requests for each process.

$request(p, t)$:
\begin{verbatim}
    \E req \in q[p]:
        /\ req.from = p
        /\ req.msg[1] = "request"
        /\ req.msg[2] = t
\end{verbatim}

$granted(p, t)$:
\begin{verbatim}
    /\ \E req \in q[p]:
            /\ req.from = p
            /\ req.msg[1] = "request"
            /\ req.msg[2] = t
            /\ \A req2 \in q[p]:
                req2[1] = "request" =>
                    \/ <<t, p>> = <<req2[2], req2[3]>>
                    \/ Less(<<t, p>>, <<req2[2], req2[3]>>)
            /\ \A p2 \in s[p]:
                \E m2 \in received[p]:
                    /\ m2.msg[1] = "ack"
                    /\ m2.msg[3] = p2
                    /\ m2.msg[2] > t
\end{verbatim}

$release(p, t)$:
\begin{verbatim}
    \E req \in rel[p]:
        /\ req.from = p
        /\ req.msg[1] = "request"
        /\ req.msg[2] = t
\end{verbatim}

The same problem arises in writing "If every \textit{requesting} process which is granted the resource eventually releases it, then every \textit{requesting process} is eventually granted", but can be remedied the same way.
QUESTION: can we prove that w/o adding the ghost var, it is not possible to even specify the liveness property?
}

\notes{
\subsection{Paxos time analysis in MIT MS thesis~\cite{de1997revisiting}}

\begin{tcolorbox}[title=Liveness property]
Let $\alpha$ be a nice execution fragment of $S_{PAX}$ starting in a reachable
state and lasting for more than $t^{i}_{\alpha} + 24l + 10nl + 13d$. Then the leader $i$ executes $Decide(v')_{i}$ by time $t^{i}_{\alpha} + 21l + 8nl + 11d$ from the beginning of $\alpha$ and at most $8n$ messages are sent. Moreover by time $t^{i}_{\alpha} + 24l + 10nl + 13d$ from the beginning of $\alpha$ any alive process $j$ executes $Decide(v')_j$ and at most $2n$ additional messages are sent.
\end{tcolorbox}

\begin{tcolorbox}[colback=blue!5!white, colframe=blue!75!black, title=Formalization]
semi-, I am counting the above definition as semi-formal
\end{tcolorbox}

\begin{tcolorbox}[colback=blue!5!white, colframe=blue!75!black, title=Proof]
Informal but rigorous proof. The proof breaks down the complete operation into intermediate milestones and measures time taken to reach each milestone. Finally all values are summed up.

Proof size: 3 paras ~ half page
Theorems used: 1
Lemmas used: 4
\end{tcolorbox}
}

\end{document}